\documentclass[a4paper,USenglish,cleveref, autoref, thm-restate]{lipics-v2021}

\title{Promptness and fairness in Muller LTL formulas} 

\author{Damien Busatto-Gaston}{Univ Paris Est Creteil, LACL, F-94010 Creteil, France}{damien.busatto-gaston@u-pec.fr}{https://orcid.org/0000-0002-7266-0927}{}
\author{Youssouf Oualhadj}{Univ Paris Est Creteil, LACL, F-94010 Creteil, France\\ CNRS, ReLaX, IRL 2000, Siruseri, India}{youssouf.oualhadj@u-pec.fr}{https://orcid.org/0000-0003-0200-4032}{}
\author{Léo Tible}{Univ Paris Est Creteil, LACL, F-94010 Creteil, France}{leo.tible@u-pec.fr}{}{}
\author{Daniele Varacca}{Univ Paris Est Creteil, LACL, F-94010 Creteil, France}{daniele.varacca@u-pec.fr}{https://orcid.org/0009-0007-6500-2153}{}

\authorrunning{D. Busatto-Gaston, Y. Oualhadj, L. Tible and D. Varacca} 

\Copyright{Damien Busatto-Gaston, Youssouf Oualhadj, Léo Tible and Daniele Varacca} 



\category{} 

\relatedversion{} 

\nolinenumbers 

\EventEditors{John Q. Open and Joan R. Access}
\EventNoEds{2}
\EventLongTitle{42nd Conference on Very Important Topics (CVIT 2016)}
\EventShortTitle{CVIT 2016}
\EventAcronym{CVIT}
\EventYear{2016}
\EventDate{December 24--27, 2016}
\EventLocation{Little Whinging, United Kingdom}
\EventLogo{}
\SeriesVolume{42}
\ArticleNo{23}

\usepackage[utf8]{inputenc}
\usepackage[T1]{fontenc}

\usepackage{amssymb} 
\usepackage[tbtags,fleqn]{amsmath} 
\usepackage[mathscr]{eucal} 

\usepackage{enumerate}

\usepackage{hyperref}

\usepackage{amsthm}
\usepackage{thmtools}
\usepackage{thm-restate} 

\usepackage{subcaption} 
\usepackage{graphicx} 

\usepackage{array}
\usepackage{multirow}
\usepackage{tabularx}
\usepackage[online]{threeparttable}

\usepackage{listings}
\usepackage{comment} 
\usepackage{xstring} 

\usepackage{xspace} 
\usepackage{xparse} 

\usepackage{amsfonts} 
\usepackage{stmaryrd} 

\usepackage{hhline} 
\usepackage[lined,boxed,commentsnumbered,noend]{algorithm2e}
\usepackage{float}

\newcommand\newmath[2]{\newcommand#1{\ensuremath{#2}\xspace}}

\newcommand\newmathope[2]{\newcommand#1{\ensuremath{\operatorname{#2}}\xspace}}

\newmath{\prob}{\bbP}

\newmath{\zug[1]}{\left\langle{#1}\right\rangle} 
\newmath{\mset[1]}{\{{#1}\}}

\newmath{\ap}{\mathsf{AP}}
\newmath{\states}{\mathsf{S}}
\newmath{\tr}{\mathsf{T}}
\newmath{\sinit}{s_{\mathsf{init}}}

\newmath{\lts}{\mathcal{S}}
\newmath{\ltsFull}{\lts = \zug{\states,\allowbreak \sinit,\allowbreak \tr,\allowbreak \lblFull}}

\newmath{\subform}{\mathsf{SubF}}

\newmath{\paths}{\mathsf{FPaths}}
\newmath{\runs}{\mathsf{Runs}}
\newmath{\runsinit}{\mathsf{Runs_{init}}}
\newmath{\cyl}{\mathsf{Cyl}}

\newmath{\attractor}{\mathsf{k}\text{-}\mathsf{Attractor}}
\newmath{\attractork[1]}{\mathsf{#1}\text{-}\mathsf{Attractor}}
\newmath{\infinity}{\mathsf{Inf}}

\newmath{\successor}{\mathsf{Succ}}
\newmath{\predecessor}{\mathsf{Pred}}
\newmath{\pre}{\mathsf{SPred}}
\newmath{\witness}{\mathsf{W}}

\newmath{\sem[1]}{\left\llbracket{#1}\right\rrbracket} 

\newmathope{\untl}{\mathsf{U}}
\newmathope{\rels}{\mathsf{R}}
\newmathope{\nxt}{\mathsf{X}}
\newmathope{\f}{\mathsf{F}}
\newmath{\p}{\mathsf{P}}
\newmathope{\fp}{\f_\p}
\newmathope{\g}{\mathsf{G}}
\newmathope{\finf}{\f^\infty}
\newmathope{\fpinf}{\f_\p^\infty}

\newmath{\fLTL}{\calL(\f^\infty)}
\newmath{\fpLTL}{\calL{(\f_{\p}^\infty)}}
\newmath{\TfpLTL}{\f\big(\fpLTLplus\big)}
\newmath{\fpLTLplus}{\calL^+(\f_{\p}^\infty)}

\newmath{\asMod}{\models_{\text{\tiny{\sf AS}}}}
\newmath{\notAsMod}{\not\models_{\text{\tiny{\sf AS}}}}

\usepackage{tikz}
\usetikzlibrary{arrows,shapes,automata,positioning}
\tikzset{>=latex,initial text=}

\tikzstyle{player0}=[state,draw,rounded rectangle,align=center,minimum size=6mm]
\tikzstyle{player1}=[state,draw,rectangle,align=center,minimum size=6mm]
\tikzstyle{player2}=[state,draw,rounded rectangle,align=center,minimum size=6mm]

\tikzstyle{action0}=[->]

\tikzset{every loop/.style={looseness=7}} 

\usetikzlibrary{decorations.pathreplacing}

\newcommand{\calL}{\mathcal{L}}

\newcommand{\bbP}{\mathbb{P}}
\newcommand{\bbN}{\mathbb{N}}

\def\set#1{\{#1\}}

\newmath{\lbl}{\mathsf{lbl}}
\newmath{\lblFull}{\lbl \colon \states \to 2^{\ap}}
\newcommand{\induction}{\mathsf{sat}}

\newtheorem{problem}[theorem]{Problem}

 \newcommand{\conp}{\mathsf{coNP}}
 \newcommand{\pspace}{\mathsf{PSPACE}}
 \newcommand{\np}{\mathsf{NP}}
 \newcommand{\poly}{\mathsf{PTIME}}
 
 \newcommand{\sub}[2]{[#1..#2]}

\renewcommand\paragraph[1]{\smallskip\noindent\textbf{#1.}}

\usepackage[disable]{todonotes} 
\usepackage[normalem]{ulem} 

\begin{document}

\maketitle

\begin{abstract}
In this paper we consider two different views of the model checking
  problem for the Linear Temporal Logic (LTL).
  On the one hand, we consider the \emph{universal} model checking
  problem for LTL, where one asks that for a given system and a given
  formula all the runs of the system satisfy the formula.  On the
  other hand, the \emph{fair} model checking problem for LTL asks that for
  a given system and a given formula almost all the runs of the system
  satisfy the formula.

  It was shown that these two problems have the same theoretical
  complexity, \textit{i.e.}~PSPACE-complete.  A less expensive fragment was identified in a previous
  work, namely the \emph{Muller fragment}, which consists of combinations of repeated reachability properties.

  We consider \emph{prompt} LTL formulas (pLTL), that extend LTL with an additional operator, \textit{i.e.}~the
  \emph{prompt-eventually}.
  This operator ensures the existence of a
  bound such that reachability properties are satisfied within this bound.
  This extension comes at no cost since the model checking problem remains PSPACE-complete.
  
  We show that the corresponding Muller fragment of pLTL, with prompt repeated reachability properties, enjoys similar computational improvements. 
  Another feature of Muller formulas is that the model checking problem becomes easier under the fairness assumption.
  This distinction is lost in the prompt setting, as we show that the two problems are equivalent instance-wise.
  Subsequently, we identify a new prefix independent fragment of pLTL for which the fair model
  checking problem is less expensive than the universal one.
\end{abstract}

\section{Introduction}

Linear Temporal Logic (LTL) allows system designers to easily describe
behavioral properties of a system~\cite{PnuAmi77}.
Its expressive power and convenience of
use proved useful in many areas such as system design and
verification~\cite{ClarkeES86, QueilleS82}, agent
planning~\cite{CGV02,GV99}, knowledge representation~\cite{KR}, and
control and synthesis~\cite{PnueliR89, PnueliR89b}.  At the heart of
these applications a fundamental formal approach is always to be found, \textit{i.e.}~the \emph{model checking problem}~\cite{pspaceLTL}.

\paragraph{Universal and fair model checking}
When trying to verify a system against a specification, one usually
models the system as directed graphs called
\emph{Labeled Transition Systems} (LTS). A run of the
system is an infinite path in the LTS.  The standard approach to
model checking consists in verifying that all possible runs of the LTS
comply with the specification.
Some systems may not satisfy the specification because of a few
\emph{unlikely} runs. To avoid discarding these systems, the
\emph{fair} model checking approach gives a formal definition of
\emph{small sets} of executions, and then verifies that the set of
executions that do not satisfy the specification is indeed small.  
\begin{figure}
  \begin{center}
  \begin{minipage}{.45\textwidth}
    \centering
    \begin{tikzpicture}[>=latex', join=bevel, thick, initial text =]
      \useasboundingbox (-1,-1) rectangle (4,1);
      \node[] (0) at (0bp, 0bp) [draw, circle, initial]{$a$};
      \node [below] at (0.south) {\sf{Idle}};
      \node[] (1) at (50bp, 0bp) [draw, circle] {$b$};
      \node [below] at (1.south) {\sf{Query}};    
      \node[] (2) at (100bp, 0bp) [draw, circle] {$c$};
      \node [below] at (2.south) {\sf{Grant}};
      \draw[->,] (0) [] to node [below] {} (1);
      \draw[->,] (0) [loop above] to node [below] {} (0);
      \draw[->,] (1) [bend right] to node [below] {} (2);
      \draw[->,] (2) [bend right] to node [below] {} (1);
      \draw[->,] (1) [loop above] to node [below] {} (1);
      \draw[->,] (2) [loop above] to node [below] {} (2);
    \end{tikzpicture}
    \end{minipage}
    \hfill
    \begin{minipage}{.45\textwidth}
    \centering
    \begin{tikzpicture}[>=latex', join=bevel, thick, initial text =]
      \useasboundingbox (-1,-1) rectangle (4,1);
      \node[] (0) at (0bp, 0bp) [draw, circle, initial]{$a$};
      \node [below] at (0.south) {\sf{Idle}};
      \node[] (1) at (50bp, 0bp) [draw, circle] {$b$};
      \node [below] at (1.south) {\sf{Query}};    
      \node[] (2) at (100bp, 0bp) [draw, circle] {$c$};
      \node [below] at (2.south) {\sf{Grant}};
      \draw[->,] (0) [] to node [above] {$\frac{1}{2}$} (1);
      \draw[->,] (0) [loop above] to node [above] {$\frac{1}{2}$} (0);
      \draw[->,] (1) [bend right] to node [below] {$\frac{2}{3}$} (2);
      \draw[->,] (2) [bend right] to node [above] {$\frac{3}{4}$} (1);
      \draw[->,] (1) [loop above] to node [above] {$\frac{1}{3}$} (1);
      \draw[->,] (2) [loop above] to node [above] {$\frac{1}{4}$} (2);
    \end{tikzpicture}
    \end{minipage}
  \end{center}
  \vspace{-5mm}
  \caption{\label{fig:simpleEx} A protocol modeled as an LTS on the left, and as a probabilistic system on the right.}
\end{figure}

\begin{example}\label{ex:intro-fair}
  Consider the example on the left of \figurename~\ref{fig:simpleEx}, this figure
  models a ``toy protocol'' that could either be in an idle state, a
  querying state or granting state. Assume now that the modeler wants
  to check that \emph{{\bf the protocol is not always idle}}. Without
  any fairness assumption, this system will be rejected since it could
  always repeat the loop on state $a$.  Now, consider that this LTS is an abstracted view of the real protocol, which is a probabilistic system modeled as a Markov chain on the right of \figurename~\ref{fig:simpleEx}.
  Then, the probability associated with the system remaining idle forever is $0$. 
  We can decide to ignore this ``unlikely'' possibility and say that the
  system \emph{fairly} satisfies the specification.  
\end{example}
  In general, a set of executions of a Markov chain is considered \emph{small} if it has probability $0$. It turns out that the precise probabilities that appear in the system have no impact on which sets of executions have probability $0$.
  Thus, we can assume without loss of generality every probability distribution to be uniform, and represent systems as LTS instead. 


\paragraph{LTL model checking}
The complexity of verifying that every run of a system satisfies an LTL formula is known to be a
$\pspace$-complete problem~\cite{pspaceLTL}.
This complexity is measured with respect to
the size of the LTS and its specification, an LTL formula in our case.
This high complexity is due to the fact that one has to
encode the specification as a
\emph{non-deterministic Büchi automaton}, that can be exponential in size.
This yields an exponential blow-up and further techniques are required to keep the complexity within $\pspace$, see \textit{e.g.}~\cite{BK08}.

In order to circumvent this blow-up, a natural idea is to identify
fragments of LTL where the model checking problem becomes easier to solve.  
In the seminal work of Sistla and Clarke~\cite{pspaceLTL}, they identified fragments
where model checking is $\conp$-complete, which is more amenable to implementation than $\pspace$ as it opens the door to symbolic approaches using SAT-solvers.  In particular they showed
this complexity for the class of formulas that consists of Boolean combinations of reachability properties. We also cite the fragment identified by Muscholl et al.~\cite{MuschollW05}. 
This fragment is obtained by prohibiting the
use of reachability operators (until) and
restricting the formula to exclusively using next operators indexed
by a letter. They showed that the model checking problem is
$\np$-complete for this fragment.  Finally we highlight the work of Emerson et
al.~\cite{EmersonL87} where they studied the fragment of \emph{Muller
  formulas}, \textit{i.e.}~the fragment combining repeated reachability (\textit{i.e.}~B\"uchi) properties.
They showed that model checking becomes $\conp$-complete for Muller formulas.

The $\pspace$-complete complexity result also holds for the \emph{fair} model checking
problem~\cite{CY95}. 
However, for the fragment of \emph{Muller Formulas}, while
the universal model checking problem becomes $\conp$-complete, the fair
model checking problem can be solved by a polynomial time algorithm presented in~\cite{SVV07}.
In other words, this fragment allows one to take advantage of the fairness assumption to obtain tractable model checking procedures.


\begin{example}
  We go back to the example of \figurename~\ref{fig:simpleEx},
  and consider the following specification:
  \emph{{\bf infinitely
      often a query is made and infinitely often a grant is granted}.}
  This specification can be expressed as a conjunction of repeated reachability objectives, and thus as a Muller formula.
\end{example}

\paragraph{Prompt Formulas}
Consider the following natural specification for messaging protocols: \textbf{All along the execution, whenever a
message is sent an acknowledgment is received at a later step}.
A system may satisfy this specification in an unpractical way, where
the waiting time for the
acknowledgment grows arbitrarily large along some executions.
\emph{Prompt LTL} (pLTL), introduced by Kupferman et al.~\cite{pspacepLTL}, can
express a variant of this specification that enforces an upper limit on waiting times: \textbf{There exists a bound $k$ such that, along any
  execution, whenever a message is sent an acknowledgment is received
  within the next $k$ steps}.

\begin{figure}
  \begin{center}
\vspace{-5mm}
    \begin{tikzpicture}[>=latex', join=bevel, thick, initial text =]
      \node[] (1) at (50bp, 0bp) [initial,draw, circle] {$a$};
      \node [below] at (1.south) {\sf{Msg}};    
      \node[] (2) at (100bp, 0bp) [draw, circle] {$b$};
      \node [below] at (2.south) {\sf{Ack}};
      \draw[->,] (1) [bend right] to node [below] {} (2);
      \draw[->,] (2) [bend right] to node [below] {} (1);
      \draw[->,] (1) [loop above] to node [below] {} (1);
      \draw[->,] (2) [loop above] to node [below] {} (2);    
    \end{tikzpicture}
  \end{center} 
  \vspace{-5mm}
  \caption{\label{fig:ltlVspltl}
  An LTS that satisfies a Muller formula but not its \emph{prompt} variant, as per Example~\ref{ex: ltlVspltl}.
  }
\vspace{-2mm}
\end{figure}

\begin{example}
  \label{ex: ltlVspltl}
  Consider the LTS of \figurename~\ref{fig:ltlVspltl}, and consider the
  specification asking that either {\sf Msg} or {\sf Ack} is seen
  infinitely often. Clearly enough, any infinite path in the system 
  will satisfy this specification.
  Consider now the prompt variant of this specification, asking for the existence of a bound $k$ such that either {\sf Msg} is seen
  infinitely often in a prompt way, \textit{i.e.}~with a maximum of $k$ steps in between successive occurrences, or {\sf Ack} is seen
  infinitely often in a prompt way.
  Now, for any bound $k\ge 0$, the run
  $a^{k+1}b^{k+1}\cdots$ does not satisfy this specification.
  As such, the system does not satisfy the prompt specification.
\end{example}

The model checking problem for pLTL
is also known to be PSPACE-complete~\cite{pspacepLTL}. This is achieved
through an \emph{efficient translation} into LTL. The fair model
checking problem has the same complexity as well: although an
explicit proof is not published, a careful inspection of the proof
of~\cite{pspacepLTL} shows that the translation into LTL formulas
holds for the fair setting, and thus it is sufficient to invoke the
algorithm from~\cite{CY95} without further blowup.

In this paper we consider the 
\emph{Muller fragment} of pLTL, that combines the prompt variants of repeated reachability properties such as the one described in Example~\ref{ex: ltlVspltl}.

\paragraph{Contributions and organisation}
Our original contributions are as follows.

We first show that universal model checking for the Muller fragment of
pLTL is $\conp$-complete. In order to show the membership, we had to
depart from the already existing reduction to classical LTL and
develop new technical tools. 
Roughly speaking, we develop combinatorial tools to represent runs
and the reasons why they might not satisfy a prompt repeated reachability property of bound $k$.
If one thinks about a faulty run as an infinite run containing a finite window of $k$ consecutive faulty states, then 
\emph{pumping} a cycle within that window leads to a longer window of faulty states.
Thus, this run can be used as a generator for faulty runs of arbitrarily long bounds $k'>k$, cf.~Lem.~\ref{lem:pboundgeneral}.
However, one has to pay particular attention to the case where temporal operators are nested. In particular, ``pumping'' these finite sequences of faulty runs is done according to a well chosen order, cf.~Section~\ref{subsec:pump} for a formal definition of the notion of \emph{multi-pumpings}.

In Section~\ref{subsec:witness}, we prove a \emph{small witness property}. This notion, in some sense, efficiently stores data about the existence of counter-examples, 
and results in a $\conp$ procedure.


Our second contribution is to show that the fair model checking problem for this fragment does not need to be studied separately as it coincides with the universal model checking problem, cf.~Thm.~\ref{thm:ufeq}.
Indeed we prove that if a
system is a fair model for some pLTL Muller formula, then every
path must satisfy the formula. 

Further, we note that prompt Muller formulas may sometimes require ``too much'' from a system. In particular it is possible for a system to violate a specification during an \emph{initial} phase of the execution, but
once it enters a
\emph{steady regime} the specification might be satisfied.

Our last contribution is to address this issue by 
introducing the \emph{initialized Muller} fragment for pLTL.
This fragment expresses the fact that a system should satisfy a
specification in the long run, \textit{i.e.}~we ignore the finite
initialization phase.  This vision is inspired from \emph{prefix
  independent} specifications. Such specification are only interested
in the asymptotic behavior of a system. For instance, \emph{parity,
  Rabin, Street, Büchi} are all $\omega$-regular specifications whose
satisfaction is independent of any finite prefix~\cite{BK08}. Not
only do these specifications seem to be more suited to real life
applications, they also in general enjoy nice properties, especially
in the probabilistic setting, cf.~\cite{Chatterjee07, GimbertH10,
  GimbertK14}.  We also mention results~\cite{BruyereHR16,Chatterjee0RR15}
where prefix independence has been considered in a
setting rather close to ours, and there again they exhibited well
behaved properties~\cite{BrihayeDOR20}.

In our case, we show that the fair model checking of prompt Muller formulas is more tractable on initialized formulas:
the universal model checking is still in $\conp$, but there exists a
polynomial time algorithm solving the fair model checking
problem for this fragment.


\begin{table}
\begin{center}
\vspace{-5mm}
\caption{\label{tbl:results}
Summary of our contributions on variants of the model checking problem.
}
\begin{tabular}{|r|cc||ccc|}
\hline
\multirow{2}{*}{Model checking} & \multicolumn{2}{c||}{non-prompt} & \multicolumn{3}{c|}{prompt} \\
 & \multicolumn{1}{c}{LTL} & Muller & \multicolumn{1}{c}{pLTL} & \multicolumn{1}{c}{Muller} & initialized Muller \\ \hline
%
%
 \multirow{2}{*}{{Universal}}      & \multicolumn{1}{c|}{ $\pspace$-c} & { $\conp$-c}   &    \multicolumn{1}{c|}{$\pspace$-c}      & \multicolumn{1}{c|}{} & {\textcolor{black}{$\conp$}}   \\ 
& \multicolumn{1}{c|}{\cite{pspaceLTL}}  & \cite{EmersonL87} & \multicolumn{1}{c|}{\cite{pspacepLTL}} & \multicolumn{1}{c|}{\textcolor{black}{$\conp$-c}} & {\textcolor{black}{Thm.~\ref{thm:TLTLnpcomp}}} \\ \cline{1-4} \cline{6-6}
\multirow{2}{*}{{Fair}}           & \multicolumn{1}{c|}{ $\pspace$-c} & { $\poly$} & \multicolumn{1}{c|}{$\pspace$-c} & \multicolumn{1}{c|}{\textcolor{black}{Thm.~\ref{thm:univconp}, \ref{thm:ufeq}}} & {\textcolor{black}{$\poly$}} \\ 
& \multicolumn{1}{c|}{\cite{CY95}} & \cite{SVV07} &  \multicolumn{1}{c|}{\cite{pspacepLTL}} & \multicolumn{1}{c|}{} & {\textcolor{black}{Thm.~\ref{thm:fairMC}}} \\ \hline
\end{tabular}
\end{center}
\vspace{-5mm}
\end{table}

\section{Preliminaries}

Throughout the document we will use the following notations and
conventions: $\ap$ is a finite set of atomic propositions. For an arbitrary
set $E$, $E^*$ is the set of finite sequences in $E$, and
$E^\omega$ is the set of infinite sequences of
$E$. When $E$ is a finite set, $|E|$ will denote its size.

\paragraph{Labelled Transition System}
An \emph{LTS} is a tuple $\ltsFull$ such that $\states$ is a set of states, $\sinit\in\states$ is an initial state, $\tr\subseteq\states\times\states$
is a set of transitions, and $\lbl\ : \ S \to 2^{\ap}$ is a  labeling function mapping every state to the atomic propositions that hold on it.

For a state $s\in\states$, the set of successors  of $s$ is
$\successor(s)=\set{t\in\states \mid (s,t)\in\tr}$.
A \emph{finite path} in $\lts$ is a finite sequence of states  $\pi=s_0 s_1 \cdots s_k$ of length  $k+1$ such that
$\forall 0\leq i\leq k-1,  s_{i+1} \in \successor(s_i)$. We denote by
$|\pi|$ the length of $\pi$, \textit{i.e.}~$|\pi| = k+1$.
A \emph{run} in $\lts$ is an infinite sequence of states $\rho = s_0 s_1 \cdots$ such that
$\forall i\geq 0$, $s_{i+1} \in \successor(s_i)$.
Let $\rho$  be a run and let $i\geq 0$, then $\rho[i]=s_i$,
$\rho[i..]=s_i s_{i+1} \cdots$, that is the infinite suffix starting
in the $(i+1)$th letter, and $\rho[..i]$ is the prefix up to the $(i)$th
position, that is $\rho[..i]=s_0\cdots s_{i-1}$.
For $i<j$, $\rho[i..j]$ is the finite path $s_i\cdots s_{j-1}$.
We use the same
notations for a finite path $\pi$, and in
this case $\pi[i..]$ will be a finite suffix.
The concatenation of a finite path $\pi$ with a finite path (or a run) $\pi'$ is denoted $\pi\pi'$. 
A cycle (or a loop) is a finite path $\pi=s_0 s_1 \cdots s_k$ such that $s_0\in\successor(s_k)$.
In particular, the finite path $\pi^n$ repeats for $n$ iterations the cycle $\pi$.

The set of states visited infinitely often by a run $\rho$ is denoted
$\infinity(\rho)$, and is formally defined as $\{s\in\states\mid \forall i\geq 0,\exists j \geq i,
\rho[j]=s\}$. 
Finally,
let $\paths$ (resp. $\runs$) be the set of all the finite paths (resp. runs) in $\lts$, and
let $\runsinit$ be the set of all runs starting from $\sinit$, that
is $\set{\rho \in \runs \mid \rho[0] = \sinit}$. 
These notations assume that $\lts$ is clear from context.

\paragraph{Linear Temporal Logic}
An \emph{LTL} formula $\varphi$ is defined using the following grammar:
\begin{align*}
  \varphi ::= \alpha \mid
  \lnot \varphi \mid
  \varphi \lor \varphi \mid
  \nxt \varphi \mid
  \varphi  \untl\varphi
  \ , \text{ where $\alpha$ ranges over $\ap$}.
\end{align*}
Runs $\rho$ of $\lts$ are evaluated inductively over LTL formulas as follows:
\begin{align*}
  \rho &\models \alpha \text{ iff } \alpha \in \lbl(\rho[0])\ ,\qquad\qquad\enspace
  \rho \models \lnot \varphi \text{ iff } \rho \not\models \varphi\ ,\\
  \rho &\models \varphi\lor\psi  \text{ iff } \rho \models \varphi \text{
    or } \rho \models \psi\ ,\qquad
  \rho \models \nxt \varphi \text{ iff } \rho[1..] \models \varphi\ ,\\
  \rho &\models \varphi \untl \psi \text{ iff } \exists i \geq 0,~
         \rho[i..] \models \psi \text{ and } \forall j < i,~
         \rho[j..] \models \varphi\ ,&&
\end{align*}

A \emph{state formula} is a formula that only contains Boolean operators ($\neg$ and $\lor$) and atomic propositions, and is thus entirely evaluated on the first position of $\rho$.
The operators $\nxt$ and $\untl$ are called \emph{temporal operators}.
We derive all standard Boolean operators from $\neg$ and $\lor$, let $\top$ and $\bot$ be atomic propositions that are respectively always true and always false, and define a few extra temporal operators as syntactic sugar:
  $\f \varphi = \top \untl \varphi\ ,
   \g \varphi = \lnot \f \lnot \varphi\ ,
  \finf \varphi = \g\f \varphi$.
The formula $\f \varphi$ is true for any run where $\varphi$ is \emph{eventually} true, while the formula $\g \varphi$ is true for any run where $\varphi$ \emph{always}
holds.
The formula $\finf \varphi$ holds for any run where
there exists infinitely many positions from where $\varphi$ is true, and is used to encode repeated reachability properties, such as the winning condition of a Büchi automaton.

\begin{problem}[Universal model checking]
  Given an LTS $\lts$ and an LTL formula $\varphi$,
  does $\rho \models \varphi$ hold true for
  every run
  $\rho \in \runsinit$?
  In this case, we write $\lts \models \varphi$.
\end{problem}

  The universal model checking problem 
  is $\pspace$-complete~\cite{pspaceLTL}.
We express the complexity of our model checking problems with respect to both the size of the formula $|\varphi|$, \textit{i.e.}~the number of operators
that appear in $\varphi$, and the size of the system $|\lts|$, \textit{i.e.}~$|\tr|+|\states|\max_{s\in\states} |\lbl(s)|$.

\paragraph{Fairness in model checking}
The fairness assumption presented in Example~\ref{ex:intro-fair} relies on the idea that a set of runs is sometimes considered quantitatively \emph{small}.
In particular, we use a \emph{fair coin} to view an LTS as a probabilistic system and derive a measure over paths. 
At each state, a fair coin is
flipped and the successor state is chosen accordingly.
This coin flipping is assumed to be i.i.d. and it induces a natural probability measure over
$\runsinit$.
In the fair model checking problem, one asks whether \emph{almost all}
the runs of an LTS satisfy a given formula $\varphi$ \textit{i.e.}~if a run obtained from the fair coin process satisfies $\varphi$ with probability $1$.

Formally, in order to build the probability measure over $\runsinit$, we use
the classical notion of \emph{cylinders}. For $\pi \in \paths$, the
cylinder induced by $\pi$, denoted $\cyl(\pi)$, is the set
$\set{\rho \in \runs \mid \pi \text{ is a prefix of } \rho}$.
Then, the probability of a run being in a cylinder $\prob_{\lts}(\cyl(\pi))$ is defined as the probability that $\pi$ is followed under the fair coin process.
This measure can be uniquely extended over the set $\runsinit$ using Carathéodory’s extension theorem. 

\begin{problem}[Fair model checking]
  Given an LTS $\lts$ and an LTL formula $\varphi$,
  does it holds that
  $\prob_{\lts}(\set{\rho \in \runsinit \mid \rho \models \varphi}) = 1$?
  In this case, we will $\lts \asMod \varphi$, where $\asMod$ stands for the ``almost sure'' satisfaction of a formula.
\end{problem}

Under the fairness assumption, a model checking procedure is intuitively allowed to ignore unrealistic behaviours. 
However, this does not simplify the complexity of the problem, as it has been shown that
  the fair model checking problem is also $\pspace$-complete~\cite{CY95}.

\paragraph{The \emph{Muller} fragment of LTL}
In an effort to obtain lower complexity results for the model checking problem, subclasses of LTL formulas, defined by syntactic restrictions, have been considered. In particular, 
    an LTL formula $\varphi$ is in the \emph{Muller fragment}, denoted $\varphi\in\fLTL$, if the repeated reachability operator $\finf$ is the only temporal operator that is allowed:
    \begin{align*}
  \varphi ::= \alpha \mid
  \lnot \varphi \mid
  \varphi \lor \varphi \mid
  \finf \varphi
  \ .
  \end{align*}

The key property of Muller formulas is that their satisfaction on a run $\rho$ is entirely determined by the initial state and $\infinity(\rho)$, the states visited infinitely often. This leads to lower complexity results: in \cite{SVV07}, it was shown that
  the universal model checking problem for Muller formulas is
  $\conp$-complete, while the fair model checking problem can be solved
  in polynomial time.

\paragraph{Promptness in  LTL}
A prompt LTL formula $\varphi$ is defined according to the following
grammar~\cite{pspacepLTL}:
$
  \varphi ::= \alpha \mid
  \lnot \alpha \mid
  \varphi \lor \varphi \mid
  \varphi \land \varphi \mid
  \nxt \varphi \mid
  \varphi  \untl\varphi \mid
  \varphi  \rels\varphi \mid
  \f_\p \varphi
$.

The main difference with LTL lies in the addition of $\f_\p$. This operator
states that $\varphi$ has to be satisfied eventually, but in a ``prompt'' fashion.  The
semantics of this operator are defined with respect to a bound $k\in\bbN$.
For a given $k$, we write $(\rho,k) \models \f_\p \varphi$ if $\exists i \leq k,~(\rho[i..],k) \models \varphi$, in contrast, recall that $\rho \models \f \varphi$ if $\exists i \in\bbN, \rho[i..] \models \varphi$.
The other operators ignore the bound $k$ and are evaluated using the semantics of LTL defined earlier. 
Another difference with LTL is that Boolean negations are not allowed in pLTL,
as the negation of the newly added 
operator $\f_\p$ is deemed unnatural from a modelling point of view.
Therefore, the grammar explicitly contains the conjunction $\land$, and
the \emph{release} operator $\rels$, the dual of the \emph{until} operator $\untl$,
previously expressed with negations.


\begin{problem}[Universal prompt model checking]
  Given an LTS $\lts$ and a pLTL formula $\varphi$, 
  does there exists a bound $k\in\bbN$ such that for all $\rho\in\runsinit$ in $\lts$,
  $(\rho,k)\models \varphi$?
  In this case, we write $\lts \models \varphi$.
\end{problem}

\begin{problem}[Fair prompt model checking]
  Given an LTS $\lts$ and a pLTL formula $\varphi$, 
  does there exists a bound $k\in\bbN$ such that $\prob_{\lts}(\set{\rho \in
  \runsinit \mid (\rho,k) \models \varphi}) = 1$?
  In this case, we write $\lts \asMod \varphi$.
\end{problem}

The addition of the prompt eventually operator $\fp$ comes at no extra cost compared with LTL, as both universal and fair model checking remain $\pspace$-complete~\cite{pspacepLTL}.

\paragraph{The prompt Muller fragment}
In this work, we consider a subclass of pLTL formulas inspired by Muller formulas.
We define a new operator, $\fpinf \varphi = \g \fp \varphi$,
as a prompt variant of $\finf$. 
Thus, $\fpinf \varphi$ holds true for a pair $(\rho, k)$ if from every position $i\in\bbN$,
a position $j\in[i,i+k]$ can be found so that $(\rho\sub{j}{},k)\models \varphi$.
A pLTL fromula $\varphi$ is in the \emph{prompt Muller fragment}, denoted $\varphi\in\fpLTL$, if it is obtained by the following grammar:
\[
  \varphi ::= \alpha \mid
  \lnot \alpha \mid
  \varphi \lor \varphi \mid
  \varphi \land \varphi \mid
  \fpinf \varphi
\,.\]

\begin{example}
    Consider the specifications mentioned in Example~\ref{ex: ltlVspltl}, that asks that
    either {\sf Msg} or {\sf Ack} is seen infinitely often.
    This is expressed by the Muller formula $\finf \mathsf{Msg} \lor \finf \mathsf{Ack}$.
    The prompt variant is the formula $\fpinf \mathsf{Msg} \lor \fpinf \mathsf{Ack}$,
    that intuitively asks that either {\sf Msg} or {\sf Ack} are seen \emph{frequently}, \textit{i.e.}~with a frequency that does not vanish along the execution.
\end{example}


\section{Prompt Muller formulas}
\label{secUniversal}
A remarkable property of the Muller fragment of LTL is that the
satisfaction of a formula only depends on the asymptotic behavior of
the system. Therefore, as shown in~\cite{SVV07} (a corollary of
a result from~\cite{EmersonL87}), a system satisfies a Muller formula
when every strongly connected set satisfies the
formula. 
This property yields an easier model checking problem, namely $\conp$-complete. In this section we focus on the prompt Muller fragment of pLTL. We show that the complexity of the model checking problem is again $\conp$-complete. However, the techniques involved are different. Indeed, they do not follow from structural properties of the transition system. We introduce combinatorial tools to establish a small witness property. We further deepen our study by considering runs obtained under the fairness assumption. We show that prompt Muller formulas describe the same set of runs with or without the fairness assumption. This differs from (non-prompt) Muller formulas, where fairness allowed for more tractable approaches.
As a corollary, we obtain that fair model checking is as expensive as universal model checking for the prompt Muller fragment.

Most of this section will be devoted to the proof of the following theorem:

\begin{restatable}{theorem}{univconp}\label{thm:univconp}
  The universal model checking problem for $\fpLTL$
  is $\conp$-complete.    
\end{restatable}


\subsection{Pumping a counter example}
\label{subsec:pump}

The first stepping stone to prove Thm.~\ref{thm:univconp} is to define the notion
of \emph{pumping}. The idea is quite simple: for a run $\rho$, a pumping of $\rho$ is a run where
some cycle of $\rho$ has been repeated. Formally, it is defined as follows.

\begin{definition}[Pumping]
  Given an LTS $\lts$ and a run $\rho\in\runs$, a \emph{pumping} of 
  $\rho$ is a run $\rho'\in\runs$ such that
  there exist $0\leq i<j$ and $l>0$, with $\rho\sub{i}{j}$ a cycle, so that $\rho'=\rho\sub{}{i}\rho\sub{i}{j}^{l-1}\rho\sub{i}{}$.
\end{definition}

If $i=0$, then by convention $\rho\sub{}{0}=\varepsilon$ is the empty path.
In particular, $\rho'$ contains $l$ copies of the cycle $\rho\sub{i}{j}$, as the last one is in $\rho\sub{i}{}$.
We will say that a pumping is a $1$-pumping, as one cycle is iterated.
A natural extension of this notion is to allow for the pumping of several cycles in a run.
Formally, we define a \emph{multi-pumping} inductively as follows, with a run being the only $0$-pumping of itself.

\begin{definition}[Multi-pumping]
  Given an LTS $\lts$, a run $\rho\in\runs$, and $n>0$, an \emph{$n$-pumping} of $\rho$ is a run $\rho'\in\runs$
  such that 
      there exist $0\leq i<j$, with $\rho\sub{i}{j}$ a cycle, and 
      there exists a $(n-1)$-pumping $\Tilde{\rho}$ of $\rho\sub{i}{}$ and some $l>0$ 
      so that
      $\rho'=\rho\sub{}{i}\rho\sub{i}{j}^{l-1}\Tilde{\rho}$. 
  A \emph{multi-pumping} of $\rho$ is a run $\rho'$ such that $\rho'$ is an $n$-pumping of $\rho$ for some $n\geq 0$.
\end{definition}

The construction of an $n$-pumping allows for several cycles of a run to be pumped, but only consecutively in a left-to-right fashion. 
This keeps the $n$ individual pumpings ordered and will allow for inductive proofs on $n$. 

  \begin{figure}[b]
\begin{center}
  \begin{minipage}{.8\textwidth}
    \centering
    \begin{tikzpicture}[thick]
    \node (A) at (-0.1,-0.05) {$\rho = $};
 \draw[-]        (0.4,0)   -- (3.7,0);
  \draw[->]        (3.7,0)   -- (4.5,0);
    \filldraw[black] (0.4,0) circle (1pt) node[anchor=south]{$a$};
    \filldraw[black] (0.8,0) circle (1pt) node[anchor=south]{$b$};
    \filldraw[black] (1.2,0) circle (1pt) node[anchor=south]{$a$};
    \filldraw[black] (1.6,0) circle (1pt) node[anchor=south]{$b$};
    \filldraw[black] (2,0) circle (1pt) node[anchor=south]{$c$};
    \filldraw[black] (2.4,0) circle (1pt) node[anchor=south]{$b$};
    \filldraw[black] (2.8,0) circle (1pt) node[anchor=south]{$a$};
    \filldraw[black] (3.2,0) circle (1pt) node[anchor=south]{$a$};
    \filldraw[black] (3.6,0) circle (1pt) node[anchor=south]{$d$};
    
  \draw[-,color=blue]        (1.2,-0.1)   -- (2.8,-0.1);
    
    \node (B) at (-0.2,-1) {$\rho' = $};
     \draw[-]        (0.4,-1)   -- (8.5,-1);
  \draw[->]        (8.5,-1)   -- (9.3,-1);
    \filldraw[black] (0.4,-1) circle (1pt) node[anchor=south]{$a$};
    \filldraw[black] (0.8,-1) circle (1pt) node[anchor=south]{$b$};
    \filldraw[black] (1.2,-1) circle (1pt) node[anchor=south]{$a$};
    \filldraw[black] (1.6,-1) circle (1pt) node[anchor=south]{$b$};
    \filldraw[black] (2,-1) circle (1pt) node[anchor=south]{$c$};
    \filldraw[black] (2.4,-1) circle (1pt) node[anchor=south]{$b$};
    \filldraw[black] (2.8,-1) circle (1pt) node[anchor=south]{$a$};
    \filldraw[black] (3.2,-1) circle (1pt) node[anchor=south]{$b$};
    \filldraw[black] (3.6,-1) circle (1pt) node[anchor=south]{$c$};
    \filldraw[black] (4,-1) circle (1pt) node[anchor=south]{$b$};
    \filldraw[black] (4.4,-1) circle (1pt) node[anchor=south]{$a$};
    \filldraw[black] (4.8,-1) circle (1pt) node[anchor=south]{$b$};
    \filldraw[black] (5.2,-1) circle (1pt) node[anchor=south]{$c$};
    \filldraw[black] (5.6,-1) circle (1pt) node[anchor=south]{$b$};
    \filldraw[black] (6,-1) circle (1pt) node[anchor=south]{$a$};
    \filldraw[black] (6.4,-1) circle (1pt) node[anchor=south]{$b$};
    \filldraw[black] (6.8,-1) circle (1pt) node[anchor=south]{$c$};
    \filldraw[black] (7.2,-1) circle (1pt) node[anchor=south]{$b$};
    \filldraw[black] (7.6,-1) circle (1pt) node[anchor=south]{$a$};
    \filldraw[black] (8,-1) circle (1pt) node[anchor=south]{$a$};
    \filldraw[black] (8.4,-1) circle (1pt) node[anchor=south]{$d$};
    
  \draw[-,color=blue]        (6,-1.1)   -- (7.6,-1.1);
  \draw[-,color=blue,dashed]        (1.2,-1.1)   -- (6,-1.1);
  \draw[-,color=green!50!black]        (6.4,-1.2)   -- (7.2,-1.2);
    
        \node (C) at (-0.2,-2) {$\rho'' = $};
     \draw[-]        (0.4,-2)   -- (10.1,-2);
  \draw[->]        (10.1,-2)   -- (10.9,-2);
    \filldraw[black] (0.4,-2) circle (1pt) node[anchor=south]{$a$};
    \filldraw[black] (0.8,-2) circle (1pt) node[anchor=south]{$b$};
    \filldraw[black] (1.2,-2) circle (1pt) node[anchor=south]{$a$};
    \filldraw[black] (1.6,-2) circle (1pt) node[anchor=south]{$b$};
    \filldraw[black] (2,-2) circle (1pt) node[anchor=south]{$c$};
    \filldraw[black] (2.4,-2) circle (1pt) node[anchor=south]{$b$};
    \filldraw[black] (2.8,-2) circle (1pt) node[anchor=south]{$a$};
    \filldraw[black] (3.2,-2) circle (1pt) node[anchor=south]{$b$};
    \filldraw[black] (3.6,-2) circle (1pt) node[anchor=south]{$c$};
    \filldraw[black] (4,-2) circle (1pt) node[anchor=south]{$b$};
    \filldraw[black] (4.4,-2) circle (1pt) node[anchor=south]{$a$};
    \filldraw[black] (4.8,-2) circle (1pt) node[anchor=south]{$b$};
    \filldraw[black] (5.2,-2) circle (1pt) node[anchor=south]{$c$};
    \filldraw[black] (5.6,-2) circle (1pt) node[anchor=south]{$b$};
    \filldraw[black] (6,-2) circle (1pt) node[anchor=south]{$a$};
    \filldraw[black] (6.4,-2) circle (1pt) node[anchor=south]{$b$};
    \filldraw[black] (6.8,-2) circle (1pt) node[anchor=south]{$c$};
    \filldraw[black] (7.2,-2) circle (1pt) node[anchor=south]{$b$};
    \filldraw[black] (7.6,-2) circle (1pt) node[anchor=south]{$c$};
    \filldraw[black] (8,-2) circle (1pt) node[anchor=south]{$b$};
    \filldraw[black] (8.4,-2) circle (1pt) node[anchor=south]{$c$};
    \filldraw[black] (8.8,-2) circle (1pt) node[anchor=south]{$b$};
    \filldraw[black] (9.2,-2) circle (1pt) node[anchor=south]{$a$};
    \filldraw[black] (9.6,-2) circle (1pt) node[anchor=south]{$a$};
    \filldraw[black] (10,-2) circle (1pt) node[anchor=south]{$d$};
    
  \draw[-,color=green!50!black,dashed]        (6.4,-2.1)   -- (8,-2.1);
  \draw[-,color=green!50!black]        (8,-2.1)   -- (8.8,-2.1);
\end{tikzpicture}
    \end{minipage}
    \hfill
    \begin{minipage}{.15\textwidth}
    \centering
    \begin{tikzpicture}[>=latex', join=bevel, thick, initial text =]
      \useasboundingbox (-1,-1) rectangle (1,1);
      \node[] (0) at (0bp, 0bp) [draw, circle, initial]{$a$};
      \node[] (1) at (0bp, 30bp) [draw, circle] {$b$};
      \node[] (2) at (-30bp, 30bp) [draw, circle] {$c$};
      \node[] (3) at (0bp, -30bp) [draw, circle] {$d$};
      \draw[->,] (0) [bend right] to (1);
      \draw[->,] (1) [bend right] to  (0);
      \draw[->,] (1) [bend right] to  (2);
      \draw[->,] (2) [bend right] to  (1);
      \draw[->,] (0) [] to  (3);
      \draw[->,] (0) [loop right] to  (0);
      \draw[->,] (3) [loop below] to  (3);
    \end{tikzpicture}
    \end{minipage}
  \end{center}
  
  \centering
    
    \caption{Graphical representation of Example~\ref{ex:multi_pumping}. Added loops are dashed.\label{fig:multi_pumping}}
\end{figure}

\begin{example}\label{ex:multi_pumping}
  Consider the run $\rho=ab(abcb)aad^\omega$ represented if \figurename~\ref{fig:multi_pumping},
  and note that $abcb$ is a cycle that can be pumped.
  Therefore, $\rho'=ab(abcb)^4aad^\omega$ is a $1$-pumping of $\rho$.
  Now, imagine that you also want to pump the cycle $bc$.
  By definition, further pumpings can only occur starting from the last ``copy'' of the cycle $abcb$. 
  In particular, $\rho''=ab(abcb)^3a(bc)^3baad^\omega$ is a $2$-pumping of $\rho$.
\end{example}

The main property of multi-pumpings is that they preserve counter-examples,
\textit{i.e.}~if a run $\rho$ does not satisfy a
prompt Muller formula $\varphi$ for some bound, then for the same bound any multi-pumping of $\rho$ will not satisfy
$\varphi$ either. 

\begin{restatable}{proposition}{pumpingnobetter}\label{prop:pumpingnobetter}
  Given an lts $\lts$, a run $\rho\in\runs$, a multi-pumping $\rho'$ of $\rho$, a formula $\varphi\in\fpLTL$ and $k\geq 0$,
  if $(\rho,k)\not\models\varphi$, then $(\rho',k)\not\models\varphi$.  
\end{restatable}

\begin{proof}[Proof scheme]
Intuitively, if there exists no nesting of $\fpinf$ operator in the formula, then the formula being false on $(\rho,k)$ only depends
on ``faulty'' windows of length $k$ that can be found in $\rho$.
As a multi-pumping only extends faulty windows by duplicating cycles, any faulty window in the original
run also exists somewhere in the multi-pumping.
If $\varphi$ has some nesting of $\fpinf$ operators, then a window being faulty or not depends on the suffix run that comes after it, which makes these considerations significantly more involved technically, as pumping changes the suffix of our runs. A full proof is detailed in Appendix~\ref{app:multipump}.
\end{proof}

The second core property of multi-pumpings, derived from Prop.~\ref{prop:pumpingnobetter}, is that they can be used to generate counter-examples for arbitrarily large bounds $k$, as long as there exists a counter example for the bound $N=|S|+1$. This is formalised as follows:

\begin{restatable}{lemma}{pboundgeneral}\label{lem:pboundgeneral}
  Let $\ltsFull$ be an LTS, $\varphi\in\fpLTL$, and let
  $N=|S|+1$.
  If there exists $\rho_N\in\runsinit$ such that $(\rho_N,N)\not\models\varphi$
  then for all $k\geq N$, there exists a multi-pumping $\rho_k$ of $\rho_N$ such that $(\rho_k,k)\not\models\varphi$.    
\end{restatable}

Thus, in order to show that a system \emph{does not satisfy} a prompt Muller formula for any bound $k$, it is enough to exhibit a single run $\rho$ so that
$(\rho,N)\not\models \varphi$. Indeed, $(\rho,k)\not\models \varphi$ is immediate for every $k\leq N$ by definition of pLTL,\footnote{Intuitively, if a faulty window exists to falsify an $\fpinf$ operator for some bound, then every shorter bound admits the same faulty window to falsify $\fpinf$.} and for any $k>N$, Lem.~\ref{lem:pboundgeneral} can be used to generate another run that will not satisfy $\varphi$ either. 

\subsection{Canonical representation for witnesses}
\label{subsec:witness}
Let $\rho$ be a run, and assume that $(\rho,N)\not\models \varphi$. Such a run witnesses that $\lts\not\models\varphi$, however, it is an infinite sequence.
Moreover, even with a
finite representation of $\rho$, checking that $(\rho,N)\not\models \varphi$ remains a challenging task.
In this section, we introduce a 
new data structure,
 that carries a representation of 
 $\rho$ and enough information to efficiently check that $(\rho,N)\not\models \varphi$.


Given an LTL formula $\varphi$, let $\subform (\varphi)$ be the set
of all subformulas of $\varphi$.

\begin{definition}\label{def:witness}
Given an LTS $\ltsFull$, a run $\rho\in\runs$, a bound $k\geq 0$ and a formula $\varphi\in\fpLTL$, a \emph{witness} for $\rho$ with bound $k$ is a function
$\witness_{\rho,k}:\subform(\varphi)\mapsto 2^\mathbb{N}$ 
that maps subformulas to finite sets and satisfies the following conditions:
\begin{itemize}
    \item $\forall i\in\witness_{\rho,k}(\alpha),\alpha\not\in\lbl(\rho[i])$
    \item $\forall i\in\witness_{\rho,k}(\lnot\alpha),\alpha\in\lbl(\rho[i])$
    \item $\forall i\in\witness_{\rho,k}(\psi_1\lor\psi_2),
    i\in\witness_{\rho,k}(\psi_1)\cap\witness_{\rho,k}(\psi_2)$
    \item $\forall i\in\witness_{\rho,k}(\psi_1\land\psi_2),
    i\in\witness_{\rho,k}(\psi_1)\cup\witness_{\rho,k}(\psi_2)$
    \item $\forall i\in\witness_{\rho,k}(\fpinf\psi),
    \exists i'\geq i,\forall i'\leq j\leq i'+k, j\in\witness_{\rho,k}(\psi)$
    \item $0\in\witness_{\rho,k}(\varphi)$
\end{itemize}
Then, we denote 
$\max(\witness_{\rho,k})=\max\{i\in\bbN\mid \psi\in\subform(\varphi) \land i\in \witness(\psi)\}$.
\end{definition}

Intuitively, such a witness justifies that $(\rho,k)\not\models\varphi$ by describing for each subformula a relevant set of positions from where they are falsified.
In particular, $0\in\witness_{\rho,k}(\varphi)$ ensures $(\rho,k)\not\models\varphi$.
Then, the other conditions ensure that inductively, the positions falsify subformulas up to reaching the atomic
propositions.



\begin{figure}
  \vspace{-5mm}
  \centering
  \begin{tikzpicture}[node distance=2cm,auto,->,>=latex,every
  node/.style={},
  level 1/.style={sibling distance=5.5cm},
  level 2/.style={sibling distance=4cm},
  level 3/.style={sibling distance=3cm},
  level 4/.style={sibling distance=1.6cm},
  level 5/.style={sibling distance=1cm},
  level 6/.style={sibling distance=0.6cm},
  level 7/.style={sibling distance=0.3cm},level distance=10mm, scale=0.95]

    \tikzstyle{TreeNode}=[draw,rounded rectangle, minimum size=7.2mm,inner sep=0]
    \tikzstyle{TreeLabel}=[color=blue]
    \tikzstyle{TreeEdge}=[solid,-]

     \node (A) at (-0.5,-0.05) {$\rho=$};
 \draw[-]        (0,0)   -- (10,0);
 \draw[->,dashed]        (10,0)   -- (11,0);
    \filldraw[black] (0,0) circle (1pt) node[anchor=south]{$0$};
    \filldraw[black] (1.5,0) circle (1pt) node[anchor=south]{$i_1$};
    \filldraw[black] (3.5,0) circle (1pt) node[anchor=south]{$i_2$};
    \filldraw[black] (4.5,0) circle (1pt) node[anchor=south]{$i_1+k$};
    \filldraw[black] (6,0) circle (1pt) node[anchor=south]{$i_3$};
    \filldraw[black] (9,0) circle (1pt) node[anchor=south]{$i_3+k$};
    
 \draw[-,color=blue]        (1.5,-0.1)   -- (3.5,-0.1);
  \draw[-,color=green!50!black]        (3.6,-0.1)   -- (4.5,-0.1);
     \node (A) at (2.5,-0.4) {$a$};
     \node  (A) at (4,-0.4) {$c$};
     
     \node[color=black]  (A) at (0,-0.7) {$\lnot \varphi$};
     
     \node[color=blue]  (A) at (2.5,-0.7) {$\lnot c$};
     \node[color=green!50!black]  (A) at (4,-0.7) {$\lnot a$};
     
    \node[color=blue]  (A) at (2.5,-1.1) {$\lnot \fpinf b$};
    
    \draw[-,color=red]        (6,-0.1)   -- (9,-0.1);
    \node[color=red]  (A) at (7.5,-1.1) {$\lnot b$};
  
  \node[TreeNode] (0) at (3,-2) {\makebox[0mm][c]{$\fpinf$}}
    child { node[TreeNode] (1) {\makebox[0mm][c]{$\land$}}
        child {node[TreeNode] (21) {\makebox[0mm][c]{$a$}} 
        edge from parent[TreeEdge] node[TreeLabel] {}
        }
        child {node[TreeNode] (22) {\makebox[0mm][c]{$\lor$}} 
            child {node[TreeNode] (31) {\makebox[0mm][c]{$\fpinf$}} 
                child {node[TreeNode] (4) {\makebox[0mm][c]{$b$}} 
                edge from parent[TreeEdge] node[TreeLabel] {}
                }
            edge from parent[TreeEdge] node[TreeLabel] {}
            }
            child {node[TreeNode] (32) {\makebox[0mm][c]{$c$}} 
            edge from parent[TreeEdge] node[TreeLabel] {}
            }
        edge from parent[TreeEdge] node[TreeLabel] {}
        }
    edge from parent[TreeEdge] node[TreeLabel] {}
    }
    ;
	\node[at=(0),xshift=-4mm,color=black] {\makebox[0mm][r]{\small$\witness_{\rho,k}:=$}};

	\node[at=(0),xshift=4mm,color=black] {\makebox[0mm][l]{\small$\mapsto\{\textcolor{black}{0}\}$}};
    \node[at=(1),xshift=4mm,color=black] {\makebox[0mm][l]{\small$\mapsto\{\textcolor{blue}{i_1,\dots,i_2,}\textcolor{green!50!black}{i_2+1,\dots,i_1+k}\}$}};
    \node[at=(21),xshift=4mm,color=black] {\makebox[0mm][l]{\small$\mapsto\{\textcolor{green!50!black}{i_2+1,\dots,i_1+k}\}$}};
    \node[at=(22),xshift=4mm,color=black] {\makebox[0mm][l]{\small$\mapsto\{\textcolor{blue}{i_1,\dots,i_2}\}$}};
    \node[at=(31),xshift=4mm,color=black] {\makebox[0mm][l]{\small$\mapsto\{\textcolor{blue}{i_1,\dots,i_2}\}$}};
    \node[at=(32),xshift=4mm,color=black] {\makebox[0mm][l]{\small$\mapsto\{\textcolor{blue}{i_1,\dots,i_2}\}$}};
    \node[at=(4),xshift=4mm,color=black] {\makebox[0mm][l]{\small$\mapsto\{\textcolor{red}{i_3,\dots,i_3+k}\}$}};

	\node[at=(32),xshift=24mm] {};
\end{tikzpicture}
  \caption{\label{fig:witnessdefinition}Example of a run that does not satisfy a formula, and the corresponding
  witness.}
  \vspace{-2mm}
\end{figure}

\begin{example}\label{ex:witnessdefinition}
  Consider the run $\rho$ of \figurename~\ref{fig:witnessdefinition}, and
  the formula $\varphi=\fpinf(a\land(\fpinf b\lor c))$. 
  Some of the labels of states in $\rho$ are represented, so that $a$ holds between positions $i_1$ and $i_2$ for example.
  The syntactic tree of the formula is also represented in \figurename~\ref{fig:witnessdefinition}, where every subtree corresponds to a subformula.
  Here, for a given bound $k$, it is assumed that $(\rho,k)\not\models\varphi$, and a potential witness for $\rho$ with bound $k$ is described on the syntactic tree of $\varphi$: every subformula is mapped to a set of positions.
  Then, one can check that every inductive rule of Def.~\ref{def:witness} holds.
  For example, for the subformula $a$, no state $\rho[i]$ such that $i\in\witness_{\rho,k}$ has label $a$,
  and for the subformula $a\land(\fpinf b\lor c)$, it holds that
  $\witness_{\rho,k}(a\land(\fpinf b\lor c))=\witness_{\rho,k}(a)\cup \witness_{\rho,k}(\fpinf b\lor c)$.
  Note how
  the witness describes which positions along $\rho$ are relevant to prove where each subformula is
  falsified, up to the position $0$ at the root $\witness_{\rho,k}(\varphi)$, so that
  $(\rho,k)\not\models\varphi$.
\end{example}

Note that Def.~\ref{def:witness} does not require a witness $\witness$ to be small in size, as it could:
\begin{itemize}
    \item map subformulas $\psi_1$ and $\psi_2$ to positions that are lost by intersection in $\witness(\psi_1\lor\psi_2)$,
    \item map $\psi$ to more than the $k$ positions needed to falsify $\witness(\fpinf \psi)$, and
    \item use positions that are needlessly far away, so that $\max(\witness)$ is too large.
\end{itemize}

We address all of these concerns by proving a \emph{small witness property}, \textit{i.e.}, there exists a path $\rho$ with $(\rho,k)\not\models\varphi$ iff there exists a \emph{small} witness for some path $\rho'$ with bound $k$.
Here, small is meant as a polynomial upper bound on the size needed to represent $\witness$, \textit{i.e.}~$|\varphi|\max(\witness)$.

\begin{restatable}{proposition}{smallwitness}\label{prop:smallwitness}
  Given a formula $\varphi\in\fpLTL$, an LTS $\ltsFull$ and a bound $k\geq 0$, there exists a run
  $\rho\in\runs$ such that $(\rho,k)\not\models\varphi$ iff there exists a finite path
  $\pi$ and a function $\witness$ such that for all runs $\pi\rho'$, $\witness$ is a witness for $\pi\rho'$
  with bound $k$ and $|\pi|\leq (k+1)|\varphi|(|\states|+1)$ and $\max(\witness)< |\pi|$.    
\end{restatable}

\begin{proof}[Proof scheme]
First, we show that if $\lts\not\models\varphi$ then a witness $\witness$ can be obtained from any counter-example run based on the semantics of pLTL.
Then, we show that if there exists a witness then there exists one
of size polynomial in $k$ and $\varphi$.
This intuitively comes from the fact that for the case $\fpinf\psi$, 
one window of length $k$ falsifying $\psi$ is sufficient, which prevents the sets of indexes from needing to be large.
Third, we show that given a witness, we can construct a finite path $\pi$ of length polynomial in the size
of the witness and the size of the system such that for any run $\rho'$, we have $(\pi\rho',k)\not\models\varphi$.
This holds because the states of $\rho$ that appear in the witness are enough to guarantee that the formula is
falsified, and in-between those states, we can always find a short path. This construction lets us obtain a polynomial bound for $\max(\witness)$, and conclude. A full proof is detailed in Appendix~\ref{app:smallwitness}.
\end{proof}

Therefore, in order to show that a system does not satisfy $\varphi$,
it is now enough to search for a finite path $\pi$ and a witness $\witness$ of polynomial size, as described in Prop.~\ref{prop:smallwitness}.

\subsection{Universal model checking}


We are now ready to prove $\conp$ membership, as a corollary of Prop.~\ref{prop:smallwitness}:

\begin{restatable}{lemma}{univinconp}\label{lem:univinconp}
  The universal model checking problem for $\fpLTL$ is in $\conp$.  
\end{restatable}

\begin{proof}
Given a formula $\varphi$ and a system $\ltsFull$, one can guess a witness of polynomial size for the bound
$N=|\states|+1$, and check that it is indeed a witness. The check can be done bottom up on the formula and is clearly
polynomial, as it consists of label checks and standard set operations.
Prop.~\ref{prop:smallwitness} guarantees that the algorithm is correct,
while Lem.~\ref{lem:pboundgeneral} ensures that
checking the bound $k$ equal to $N$ is equivalent to checking every bound in $\bbN$.
This non-deterministic procedure is detailed in 
Algorithm~\ref{alg:promptunivMC}.
\end{proof}

\begin{algorithm}
 \caption{$\conp$~algorithm for the universal model checking of $\fpLTL$.\label{alg:promptunivMC}}
 \KwData{An lts $\ltsFull$ and a formula $\varphi\in\fpLTL$}
 \KwResult{whether $\lts\not\models \varphi$}
 \SetKwFunction{fun}{CheckW}
 \SetKwFunction{funpath}{CheckPath}
 \medskip
 \SetKwProg{myalg}{GuessAndCheck}{}{}
 \SetKwProg{check}{Check}{}{}
 \SetKwData{n}{N}
 \myalg{$(\lts, \varphi)$}{
   guess a finite path $\pi$ such that $|\pi|\leq (|\states|+2)|\varphi|(|\states|+1)$\;
   guess a function $\witness:\subform(\varphi)\mapsto 2^\mathbb{N}$ such that $\max(\witness)< |\pi|$\;
   \Return $(0\in\witness(\varphi))\land$\fun{$\lts, \witness, \varphi,\pi$}\;
 }
 \SetKwProg{myFun}{}{}{}
 \myFun{\fun{$\lts, \witness, \varphi,\pi$}}{
   \If{$\varphi=\alpha$} {
     \Return $\forall i\in\witness(\varphi), a\not\in\lbl(\pi[i])$\;
   }
   \ElseIf{$\varphi=\lnot\alpha$} {
     \Return $\forall i\in\witness(\varphi), a\in\lbl(\pi[i])$\;
   }
   \ElseIf{$\varphi=\psi_1\lor\psi_2$}{
     \Return $(\witness(\varphi)=\witness(\psi_1)\cap\witness(\psi_2))\land$\fun{$\lts,\witness, \psi_{1},\pi$} $\land$ \fun{$\lts,\witness, \psi_{2},\pi$}\;
   }
   \ElseIf{$\varphi=\psi_1\land\psi_2$}{
     \Return $(\witness(\varphi)=\witness(\psi_1)\cup\witness(\psi_2))\land$\fun{$\lts,\witness, \psi_{1},\pi$} $\land$ \fun{$\lts,\witness, \psi_{2},\pi$}\;
   }
   \ElseIf{$\varphi=\fpinf\psi$}{
    \Return $(\exists 0\le i \le |\pi|, \forall i\le j \le i+|\states|+1, j\in\witness(\psi))\land$ \fun{$\lts,\witness,\psi,\pi$}\;
   }
 }

\end{algorithm}

In order to finish the proof of Thm.~\ref{thm:univconp}, we show the following lower complexity bound.

\begin{restatable}{lemma}{flathardness}\label{lem:flathardness}
  The universal model checking problem for $\fpLTL$ is $\conp$-hard.  
\end{restatable}

This result is obtained as a reduction from the Boolean satisfiability problem to model checking, based on \figurename~\ref{fig:hardness}.
Somewhat classically, the reduction draws parallels between executions in this system and valuations over the Boolean variables $x_1,\dots,x_n$, based on which states $x_i$ are visited.
The main novelty resides in the $\fpLTL$ formula built in the reduction: as we cannot use reachability properties directly, we need to carefully encode ``reaching $x_i$'' in a roundabout way, that uses $\fpinf$ operators and the self-loops on each $x_i$.

\begin{figure}
\vspace{-5mm}
\centering
\begin{tikzpicture}[>=latex', join=bevel, thick, initial text =,scale=0.9]
    \node[state] (0) at (0bp, 0bp) [draw, circle, initial]{$\sinit$};
    \node[state] (1) at (60bp, 20bp) [draw, circle] {$x_1$};
    \node[state] (2) at (60bp, -20bp) [draw, circle] {$\overline{x_1}$};
    \node[state] (3) at (120bp, 20bp) [draw, circle] {$x_2$};
    \node[state] (4) at (120bp, -20bp) [draw, circle] {$\overline{x_2}$};
    \node[state] (5) at (180bp, 0bp) [draw=none] {$\cdots$};
    \node[state] (6) at (180bp, 20bp) [draw=none] {};
    \node[state] (7) at (180bp, -20bp) [draw=none] {};
    \node[state] (8) at (240bp, 20bp) [draw, circle] {$x_n$};
    \node[state] (9) at (240bp, -20bp) [draw, circle] {$\overline{x_n}$};
    \node[state] (10) at (300bp, 0bp) [draw, circle] {$s_n$};
    \draw[->,] (0) [] to node [below] {} (1);
    \draw[->,] (0) [] to node [below] {} (2);
    \draw[->,] (2) [loop below] to node [below] {} (2);
    \draw[->,] (1) [loop above] to node [below] {} (1);
    \draw[->,] (4) [loop below] to node [below] {} (4);
    \draw[->,] (3) [loop above] to node [below] {} (3);
    \draw[->,] (1) [] to node [below] {} (3);
    \draw[->,] (1) [] to node [below] {} (4);
    \draw[->,] (2) [] to node [below] {} (3);
    \draw[->,] (2) [] to node [below] {} (4);
    \draw[->,] (7) [] to node [below] {} (8);
    \draw[->,] (7) [] to node [below] {} (9);
    \draw[->,] (6) [] to node [below] {} (8);
    \draw[->,] (6) [] to node [below] {} (9);
    \draw[->,] (3) [] to node [below] {} (6);
    \draw[->,] (3) [] to node [below] {} (7);
    \draw[->,] (4) [] to node [below] {} (6);
    \draw[->,] (4) [] to node [below] {} (7);
    \draw[->,] (8) [] to node [below] {} (10);
    \draw[->,] (9) [] to node [below] {} (10);
    \draw[->,] (9) [loop below] to node [below] {} (9);
    \draw[->,] (8) [loop above] to node [below] {} (8); 
    \draw[->,] (10) [loop above] to node [below] {} (10); 
\end{tikzpicture}
\vspace{-3mm}
\caption{\label{fig:hardness} System used in Lem.~\ref{lem:flathardness}. States are labeled by their name, \textit{e.g.}~$x_1\in\lbl(x_1)$.}
\vspace{-3mm}
\end{figure}

\begin{proof}
    Let $\bigwedge\limits_{i=1}^l \bigvee\limits_{j=1}^3 l_{i,j}$ be an instance of 3-SAT over variables $x_1,\dots,x_n$, where every literal $l_{i,j}$ is either a variable $x$ or its negation $\bar{x}$.
    Consider the system $\lts$ in \figurename~\ref{fig:hardness}, and the $\fpLTL$ formula $\varphi=\varphi_1\lor\varphi_2$, obtained from $
  \varphi_1=\bigvee\limits_{i=1}^n(\fpinf \lnot x_i \land \fpinf \lnot \bar{x_i})$ and $
  \varphi_2=
  \bigvee\limits_{i=1}^l \bigwedge\limits_{j=1}^3 \fpinf\lnot l_{i,j}
  $.
  We show that the 3-SAT formula is \emph{unsatisfiable} (a $\conp$-hard problem) iff $\lts\models\varphi$.

  Assume that for every valuation $\nu$, $\nu\not\models\bigwedge\limits_{i=1}^l \bigvee\limits_{j=1}^3 l_{i,j}$, and thus $\nu\models\bigvee\limits_{i=1}^l \bigwedge\limits_{j=1}^3 \bar{l}_{i,j}$.
  Every path from $\sinit$ to $s_n$ can be seen as encoding a valuation $\nu$, based on the visited variables.
  Then, for every run $\rho$ in $\lts$ that reaches $s_n$,
  we have $\rho\models \bigvee\limits_{i=1}^l \bigwedge\limits_{j=1}^3 \f \bar{l}_{i,j}$.
  Note that a run that does not reach $s_n$ (and gets stuck in one of the loops) also satisfies this formula, as removing $\f$ terms from conjunctions can only help.
  Moreover, we note that $\rho\models \f \bar{l}_{i,j}$ implies $\rho\models \g\neg l_{i,j}$, as a run of $\lts$ cannot visit both a variable and its negation.
  Then, by definition of $\fpinf$ as $\g\fp$, we have that $\rho\models \g\neg l_{i,j}$ is equivalent with $(\rho,0)\models \fpinf \neg l_{i,j}$.
  Therefore, for every run $\rho$, $(\rho,0)\models\varphi_2$.
  This implies that there exists a $k$ such that every run satisfies $\varphi_1\lor\varphi_2$, \textit{i.e.}~$\lts\models\varphi$.

  Assume now that $\lts\models\varphi$, \textit{i.e.}~there exists a $k$ so that every run satisfies $\varphi_1\lor\varphi_2$.
  Let $R$ be the set of runs that go through $\lts$ by iterating every self-loop on the states $x_i$ or $\bar{x_i}$ exactly $k+1$ times before continuing, until $s_n$ is reached.
  Then, for every run $\rho\in R$ and every variable $x$, we have 
  that either $(\rho,k)\models\fpinf\lnot x$ holds (if $x$ is not visited), or $(\rho,k)\models\fpinf\lnot \bar{x}$ holds (if $\bar{x}$ is not visited).
  This is an exclusive either/or because every state that is seen is iterated $k+1$ times, so that $(\rho,k)\not\models\varphi_1$.
  It follows that the runs in $R$ must all satisfy $\varphi_2$.
  Since a run in $R$ satisfies $\fpinf\lnot x$ iff
  it satisfies $\f \bar{x}$, every run
  in $R$ satisfies $\bigvee\limits_{i=1}^l \bigwedge\limits_{j=1}^3 \f\bar{l}_{i,j}$.
  Then, if we interpret the runs in $R$ as valuations $\nu$ based on which variables are visited,
  we have $\nu\models\bigvee\limits_{i=1}^l \bigwedge\limits_{j=1}^3 \bar{l}_{i,j}$ for every valuations $\nu$, so that the 3-SAT formula $\bigwedge\limits_{i=1}^l \bigvee\limits_{j=1}^3 l_{i,j}$ is unsatisfiable.
\end{proof}

\subsection{Expressiveness under the fairness assumption}
\label{secFair}
Let us focus our attention on the expressiveness of the prompt Muller fragment under fairness. It turns out that for this fragment, assuming fairness does not change the interpretation of a formula. In particular, the complexity of the model checking problem remains the same.

Indeed, we note the following property of $\fpLTL$, obtained as a corollary of Prop.~\ref{prop:smallwitness}:
\begin{corollary}\label{cor:ufeq}
  Given a formula $\varphi\in\fpLTL$, an LTS $\lts$ and a bound $k\geq 0$, 
  if there exists a run $\rho$ so that
  $(\rho,k)\not\models\varphi$, then there exists a finite path $\pi$ 
  such that for all $\rho'\in\cyl(\pi)$, 
  $(\rho',k)\not\models\varphi$.
\end{corollary}

Thus, $\lts\not\models \varphi$ implies for any choice of $k$ the existence of an entire cylinder $\cyl(\pi)$
of counterexamples for the bound $k$.
The cylinder of a finite path must always have non-zero measure under $\prob_\lts$, by definition of our coin-toss process.
As such, $\lts\not\models \varphi$ implies $\prob_{\lts}(\set{\rho) \in \runsinit \mid (\rho,k) \models \varphi}) < 1$
for every $k$, so that $\lts\notAsMod\varphi$. This proves the following theorem, as the reverse implication holds by definition.

\begin{theorem}
  \label{thm:ufeq}
  For all LTS $\lts$ and every formula $\varphi\in\fpLTL$,
  $\lts\models \varphi$ iff $\lts \asMod \varphi$.
\end{theorem}

Let us give some alternative intuition as to why the complexity drop between the universal and fair model checking problems
for Muller formulas does not extend under promptness:
\begin{itemize}
\item The study of an $\fLTL$ formula in a system under the fairness assumption
is based on the core principle that the only thing that matters is the ``bottom SCC'' the run ends up in,
as $\infinity(\rho)$ is almost always equal to such a component.
A $\poly$ algorithm is derived from this principle~\cite{SVV07}.
Crucially, this approach exploits the fact that a formula of the shape $\finf \varphi$ is \emph{prefix independent}, in the sense that any deviation made within a finite prefix of a run $\rho$ will not change $\infinity(\rho)$, and thus will not change its satisfaction of $\finf \varphi$.

\item In the prompt setting however, a formula of the shape $\fpinf \varphi$ is clearly \emph{not prefix independent},
as its satisfaction for a bound $k$ is impacted by whether finite prefixes contain faulty windows of length $k$ or not.
Thus, the fairness assumption does not allow us to restrict the analysis of prompt Muller formulas to bottom SCCs in the same way.
\end{itemize}


\section{Initialized systems}
\label{secInitSys}
We now reflect on the expressiveness of the prompt Muller fragment. 
The motivation to replace $\finf \varphi$ by $\fpinf \varphi$ was to reinforce the guarantees obtained by executions of the system:
By enforcing a strong notion of regularity in the occurrences of $\varphi$, we prevent a good event $\varphi$ from being seen infinitely often but with a vanishing frequency. 
Indeed, requiring $\fpinf \varphi$ is sufficient to imply a frequency for $\varphi$ of at least $1/k$ for some $k$.

We now argue that this requirement may be ``too strong'', as some systems may be rejected despite enforcing this kind of non-zero frequency guarantee on the occurrences of $\varphi$.
\begin{figure}[b]
    \begin{center}
  \begin{minipage}{.45\textwidth}
    \centering
    \vspace{-3mm}
      \begin{tikzpicture}[>=latex', join=bevel, thick, initial text =]
    \node[] (0) at (0bp, 0bp) [draw, circle, initial]{$a$};
    \node [below] at (0.south) {$A$};
    \node[] (1) at (60bp, 0bp) [draw, circle] {$b$};
    \node [below] at (1.south) {$B$};
    \draw[->,] (0) [] to node [below] {} (1);
    \draw[->,] (0) [loop above] to node [below] {} (0);
    \draw[->,] (1) [loop above] to node [below] {} (1);
  \end{tikzpicture}
    \end{minipage}
    \hfill
    \begin{minipage}{.45\textwidth}
    \centering
      \begin{tikzpicture}[>=latex', join=bevel, thick, initial text =]
    \node[] (1) at (0bp, 0bp) [draw, circle,initial] {$a$};
    \node [below] at (1.south) {$A$};    
    \node[] (2) at (60bp, 0bp) [draw, circle] {$b$};
    \node [below] at (2.south) {$B$};
    \draw[->,] (1) [bend right] to node [below] {} (2);
    \draw[->,] (2) [bend right] to node [below] {} (1);
    \draw[->,] (1) [loop above] to node [below] {} (1);
    \draw[->,] (2) [loop above] to node [below] {} (2);    
  \end{tikzpicture}
    \end{minipage}
  \end{center}
  \vspace{-3mm}
  \caption{\label{fig:exampleprefind} Two systems that satisfy $\finf A\lor \finf B$.
  The one on the left does not satisfy $\fpinf A\lor \fpinf B$ unless its ``initialization period'' is ignored. The rightmost one does not satisfy it either ways.}
\end{figure}

\begin{example}
 Consider the system presented on the left in \figurename~\ref{fig:exampleprefind}, of atomic propositions $a$ and $b$, and the prompt Muller formula $\varphi=\fpinf A\lor \fpinf B$.
 Despite satisfying the Muller formula $\finf A\lor \finf B$ (every run will either stay in $a$ forever or jump to $b$ and stay there forever), this system does not satisfy the prompt variant $\varphi$.
 However, in every run either $a$ or $b$ happen with a long-term average frequency of $1$.
 the key difference here, once again, is that the long-term average frequency of an event is a prefix independent notion, unlike $\varphi$.
\end{example}

\subsection{Towards prefix independence}
We argue that being unaffected by the addition of a prefix is a desirable property for a specification: in practice, a system might require a ``guarantee-less'' initialization period before reaching a steady regime where stronger conditions can be enforced.
We introduce a \emph{new fragment} of pLTL, capturing specifications that allow the system to pass the initialization period while enforcing prompt Muller guarantees on the regularity of good events eventually.

\noindent The intuition for this construction is based on the following reasoning:
\begin{itemize}
\item Coro.~\ref{cor:ufeq} means that the satisfaction of a formula in
$\fpLTL$ can always be proven false because of some finite path.
From that point of view, prompt Muller formulas behave like safety conditions $\g \alpha$,
as opposed to Muller formulas that behave like repeated reachability objectives.
This is the crux of what prevents prefix independence.
\item In order to allow a safety formula $\g \alpha$ to accommodate for an initialisation period in the system, a natural idea is to replace it with the formula $\f \g \alpha$, that is prefix independent.
\item Following this intuition, we introduce the \emph{initialized} variant of $\fpLTL$, that contains formulas of the shape $\f \varphi$, with $\varphi\in\fpLTL$.
\end{itemize}

\begin{example}
 Consider again the system on the left of \figurename~\ref{fig:exampleprefind}, and the prompt Muller formula $\varphi=\fpinf A\lor \fpinf B$.
 Despite not satisfying $\varphi$, this system does satisfy $\f\varphi$, as in every run $\rho$ either $\g a$ or $\g b$ eventually holds, and thus $(\rho,0)\models \f(\fpinf A\lor \fpinf B)$.
 Consider now the system on the right of \figurename~\ref{fig:exampleprefind}. There, the run $abaabbaaabbb\dots$ witnesses that neither $\varphi$ nor $\f\varphi$ are satisfied.
\end{example}

We must finally address a last point of detail. Even without promptness, a Muller formula in $\fLTL$ may not be prefix-independent: indeed, a state formula with no $\finf$ operators solely depends on the initial state. In more general terms, the presence of atomic propositions outside of the scope of an $\finf$ prevents prefix independence, and must be forbidden.
This is without loss of generality for $\fLTL$ formulas, as replacing such atomic propositions by true or false depending on the initial state of the system can be done as a pre-processing step~\cite{VV12}.

Formally, we introduce $\fpLTLplus$ as formulas where every atomic proposition is in the scope of an $\fpinf$ operator. They are generated by the nonterminal $\varphi$ in the following grammar: \begin{align*}
  \varphi ::= &
  \varphi \lor \varphi \mid
  \varphi \land \varphi \mid
  \fpinf \psi\\
  \psi ::= & \alpha \mid
  \lnot \alpha \mid
  \psi \lor \psi \mid
  \psi \land \psi \mid
  \fpinf \psi
  \ .
\end{align*}

For example, $\fpinf\sf{A}$ belongs to $\fpLTLplus$, but $\sf{B} \lor\fpinf\sf{A}$ does not.

\begin{definition}
  The \emph{initialized} fragment $\TfpLTL$ is defined as
  $\{\f \varphi \mid \varphi\in\fpLTLplus\}$.
\end{definition}

This fragment enjoys the property of prefix independence:

\begin{proposition}
If $\rho$ is a run and $\f\varphi$ is a formula in $\TfpLTL$, then 
for any $k\ge 0$ and $i\ge 0$, $(\rho,k) \models \f\varphi$ iff $(\rho[i..],k) \models \f\varphi$.
\end{proposition}

\subsection{Universal model checking}


    In this section, we show that the universal model checking problem remains in $\conp$ for $\TfpLTL$,
    by adapting the techniques we developed for $\fpLTL$.
    However, we note that the lower bound is lost in the process, as the reduction detailed in Lemma~\ref{lem:flathardness} does not play well with $\f \varphi$ formulas.    

\begin{restatable}{theorem}{TLTLnpcomp}\label{thm:TLTLnpcomp}
  The universal model checking problem for $\TfpLTL$ is in $\conp$.  
\end{restatable}
The idea is to make use of the same pumpings and witness structures as Section~\ref{secUniversal} to build a short witness for $\lts\not\models \f \varphi$, while
taking into account the fact that only long-term behaviours of the system should matter.
This means that a ``faulty window'' where a subformula $\psi$ does not hold can only serve as counter-example to $\f (\fpinf \psi)$ if it can be reached after an arbitrarily long prefix.
Intuitively, faulty windows that are reached infinitely often along a run fit that description, as any suffix of the run will eventually reach them.
We show that these are the only witnesses that we need, and that they remain small in size:
\begin{restatable}{proposition}{smallwitnessinitialised}\label{prop:smallwitnessinitialised}
  Given an LTS $\ltsFull$, a formula $\f\varphi\in\f(\fpLTLplus)$ and a bound $k\geq 0$, there exists a run
  $\rho\in\runs$ such that $(\rho,k)\not\models\f\varphi$ iff there exists a finite path
  $\pi$ and a function $\witness$ such that for all run $\pi\rho'$, $\witness$ is a witness for $\pi\rho'$ and $k$, $|\pi|\leq (k+1)|\varphi|(|\states|+1)$, $\max(\witness)< |\pi|$ and $\pi[0]$ is reachable
  from $\rho[0]$ and $\pi[|\pi|-1]$.    
\end{restatable}
This statement is very close to Prop.~\ref{prop:smallwitness},
the only difference is the last condition, that is $\pi[0]$ is reachable
  from $\rho[0]$ and $\pi[|\pi|-1]$, so that the finite
path $\pi$ can be repeated infinitely often, as part of a lasso $\pi_0(\pi\pi_1)^\omega$ for some finite paths $\pi_0$ and $\pi_1$.

\begin{proof}[Proof scheme]
    If $(\rho,k)\not\models\f\varphi$, then for every suffix of $\rho$ we can
 apply Prop.~\ref{prop:smallwitness} to get a short witness of $(\rho\sub{i}{},k)\not\models\varphi$.
 Note that these paths $\pi$ and mappings $\witness$ are all bounded in size by the same bounds on $k$ and $|\varphi|$,
 so that there are finitely many of them. 
 Eventually, we must visit the same state at two positions $i$ and $j$ far enough apart to enforce that the witness path $\pi$ for $(\rho\sub{i}{},k)\not\models\varphi$ ends before $j$.
 Thus, we get $\pi[0]$ is reachable from $\rho[0]$ and $\pi[|\pi|-1]$.
 
  The converse direction of the proof is straight-forward, as we can deduce from such a lasso-shaped witness run that $(\rho\sub{i}{},k)\not\models\varphi$ for infinitely many positions $i$.
  This implies $(\rho,k)\not\models\f\varphi$ by prefix-independence.
\end{proof}

The proof of Thm.~\ref{thm:TLTLnpcomp} is then obtained from
Prop.~\ref{prop:smallwitnessinitialised}.
More precisely, 
we use a variation of Algorithm~\ref{alg:promptunivMC},
that 
guesses $\pi$ and $\witness$ and
adds the extra reachability checks of Prop.~\ref{prop:smallwitnessinitialised}
over the end-points of $\pi$.
Notice that this extra step is polynomial 
and therefore is not
detrimental to membership in $\conp$.

\subsection{Fair model checking}

Going back to the bigger picture, we recall that in the non-prompt setting a polynomial procedure is obtained for the Muller fragment under fairness.
This is based on exploiting the prefix-independence property
on the one side, and the propensity of almost all runs to maximise the strongly connected sets of states they visit infinitely often on the other side.
As such, only runs visiting bottom SCCs (maximal strongly connected sets of states almost always reached by runs in probabilistic systems) are relevant under fairness.

The intuition is then that every such SCC is either entirely winning for a prefix-independent $\fLTL$ objective $\varphi$, regardless of what initial state is picked, or entirely losing.
Moreover, checking this fact for a given bottom SCC and a given formula is straightforward, as 
$\varphi$ satisfaction is entirely determined by $\infinity(\rho)$ on a given run $\rho$, and $\infinity(\rho)$ is almost always equal to the full bottom SCC by fairness.

For a $\TfpLTL$ formula, we follow the same recipe: bottom SCCs will also be eventually reached, and will either be entirely winning for a given $\varphi$, or entirely losing for $\varphi$. 
The main difference brought by promptness is that we can no longer simply rely on fairness to make sure that every state of the SCC is visited infinitely often: in order to guarantee \emph{prompt} visits,
we must assume the fair coin to be adversarial and make sure that every state of the SCC \emph{will be visited shortly} no matter what.
This results in a kind of attractor computation to check if a bottom SCC is winning or not,
that can then be paired with the algorithm from $\fLTL$.
This is the key to proving that indeed, 
prefix independence makes the fair model checking problem
for $\TfpLTL$ polynomial.


\begin{restatable}{theorem}{thmfairMC}\label{thm:fairMC}
  The fair model checking problem for $\TfpLTL$ 
  is in $\poly$.
\end{restatable}

\section{Discussion and conclusion}
\label{sec:discussion}
In this section we discuss our results, some immediate corollaries and
future work.  The first point we address might seem technical at a
first glance, but we believe that it is worth mentioning; it concerns the
quantification of the bound $k$ used in the semantics of
pLTL.
\paragraph{Strong against weak semantics}
Recall that pLTL formulas are evaluated with respect to a uniform
bound $k$, \textit{i.e.}~for a system to be a model to some formula, all its
runs are evaluated using the same bound $k$. Let us refer to this
semantics as \emph{strong}. One could want to relax this semantics and
define the following semantics: A system $\lts$ satisfies a formula
$\varphi$ if for every run $\rho$ there exists a bound $k$ such that
$(\rho,k)\models \varphi$. Let us call this semantics \emph{weak}. Indeed
one could argue that the strong semantics is too conservative and that
a system designer might be interested in a more permissive
behavior. This raises the following questions: \emph{are both
  semantics equivalent?}  Obviously, the strong semantics always
implies the weak one.  However, the converse does not hold.  To
separate these two semantics, consider the system depicted on the left of
Figure~\ref{fig:exampleprefind}, and the formula
$\varphi=\fpinf \sf{A}\lor\fpinf \sf{B}$.
This system satisfies the formula $\varphi$ with respect to the
weak semantics but not the strong one. This raises in turn another
natural question: \emph{Are these two semantics equivalent in some
  fragment?} We answer positively to this question by noticing that they
collapse for formulas in $\f(\fpLTLplus)$. This follows from
the following observation:
\begin{center}
  \emph{The weak and strong semantics are equivalent for formulas in
    $\varphi\in\fpLTL$ if they are evaluated over a strongly
    connected LTS}.
\end{center}
Indeed, if for every $k$ there exists a counter-example run $\rho_k$ so that $(\rho_k,k)\not\models\varphi$, then by our small witness property (Prop.~\ref{prop:smallwitness}) we can get paths $\pi_1,\pi_2,\dots$ that are sufficient to falsify $\varphi$ for $k=1,2,\dots$ respectively.
It is then sufficient to chain them one after the other -- by exploiting the strongly connected assumption -- to construct a single path $\rho$ that falsifies every bound $k$, as required by the weak semantics.
As a consequence, both semantics are equivalent in our prefix independent fragment $\TfpLTL$, where satisfaction for a run $\rho$ is determined by the long-term behavior of $\rho$, \textit{i.e.}~the suffixes that stay confined to the strongly connected set $\infinity(\rho)$.

\paragraph{Probabilistic model checking}
The last point we want to discuss is the complexity of quantitative
verification for the fragment $\f(\fpLTLplus)$. Once again, our
fragment behaves well. We argue that one can compute the \emph{satisfaction
  probability} of a system with respect to a formula in
$\f(\fpLTLplus)$ in polynomial time.
This nice property follows from the following \emph{zero-one-law} for the
bottom SCCs of an LTS $\lts$:
\begin{center}
  \emph{Let $B$ be a bottom SCC and $\varphi$ be a formula in $\fpLTLplus$, then
    either almost all the runs of $B$ satisfy $\varphi$ or almost no run
    of $B$ does.}
\end{center}
This is a consequence of Coro.~\ref{cor:ufeq}, as any finite path that witnesses the falsification of $\varphi$ in a bottom SCC will eventually be visited with probability 1.
Now to conclude, one has to notice that computing the probability of
reaching $B$ can be done in $\poly$, and use the prefix
independence of $\f(\fpLTLplus)$.

\paragraph{Future work} Finally, we hint at some future
directions. 
The fragment $\f(\fpLTLplus)$ turned out to behave well in the presence of fairness, leading to a tractable model checking procedure.
While, this is an improvement over the $\conp$ procedure for the  universal problem, we are still missing a matching lower bound to separate the two formally.
A more ambitious perspective is the study of the
controller synthesis problem induced by these fragments in 2-player games. 


\bibliography{biblio}

\newpage

\appendix


\section{Multi-pumpings and their properties}\label{app:multipump}

\pumpingnobetter*

In order to prove Prop.~\ref{prop:pumpingnobetter}.
we establish some structural properties of multi-pumpings. 

A multi-pumping of a run can be seen as pumping multiple loops of the same run, 
but with an ordering of the loops first. Indeed, loops are pumped in order, one after the other, and not one inside another.
This implies that merging two multi-pumpings is quite easy.
It suffices to respect the order of the loops, as illustrated in \figurename~\ref{fig:mergepumping}.
This is formalise by the following lemma.

  \begin{figure}[h!]
  \centering
    
\begin{tikzpicture}[thick]
    \node (A) at (-0.5,-0.05) {$\rho=$};
 \draw[-]        (0,0)   -- (5,0);
 \draw[->]        (5,0)   -- (6,0);
    \filldraw[black] (1,0) circle (1pt) node[anchor=south]{};
    \filldraw[black] (1.5,0) circle (1pt) node[anchor=south]{};
    \filldraw[black] (2,0) circle (1pt) node[anchor=south]{};
    \filldraw[black] (3,0) circle (1pt) node[anchor=south]{};
    \filldraw[black] (4,0) circle (1pt) node[anchor=south]{};
    \filldraw[black] (4.5,0) circle (1pt) node[anchor=south]{};
  
 \draw[-,color=red]        (1,-0.1)   -- (2,-0.1);  
 \draw[-,color=green!50!black]        (1.5,-0.2)   -- (3,-0.2);  
 \draw[-,color=blue]        (4,-0.1)   -- (4.5,-0.1);
    
    \node (B) at (-0.5,-1.05) {$\rho_1=$};
 \draw[-]        (0,-1)   -- (7.5,-1);
 \draw[->]        (7.5,-1)   -- (8.5,-1);
    \filldraw[black] (1,-1) circle (1pt) node[anchor=south]{};
    \filldraw[black] (2,-1) circle (1pt) node[anchor=south]{};
    \filldraw[black] (3,-1) circle (1pt) node[anchor=south]{};
    \filldraw[black] (3.5,-1) circle (1pt) node[anchor=south]{};
    \filldraw[black] (4,-1) circle (1pt) node[anchor=south]{};
    \filldraw[black] (5,-1) circle (1pt) node[anchor=south]{};
    \filldraw[black] (6,-1) circle (1pt) node[anchor=south]{};
    \filldraw[black] (6.5,-1) circle (1pt) node[anchor=south]{};
    \filldraw[black] (7,-1) circle (1pt) node[anchor=south]{};
    
 \draw[-,color=red,dashed]        (1,-1.1)   -- (3,-1.1);  
 \draw[-,color=red]        (3,-1.1)   -- (4,-1.1);  
 \draw[-,color=green!50!black]        (3.5,-1.2)   -- (5,-1.2);  
 \draw[-,color=blue,dashed]        (6,-1.1)   -- (6.5,-1.1); 
 \draw[-,color=blue]        (6.5,-1.1)   -- (7,-1.1);
 
     \node (C) at (-0.5,-2.05) {$\rho_2=$};
 \draw[-]        (0,-2)   -- (7.5,-2);
 \draw[->]        (7.5,-2)   -- (8.5,-2);
    \filldraw[black] (1,-2) circle (1pt) node[anchor=south]{};
    \filldraw[black] (1.5,-2) circle (1pt) node[anchor=south]{};
    \filldraw[black] (2,-2) circle (1pt) node[anchor=south]{};
    \filldraw[black] (3,-2) circle (1pt) node[anchor=south]{};
    \filldraw[black] (4.5,-2) circle (1pt) node[anchor=south]{};
    \filldraw[black] (5.5,-2) circle (1pt) node[anchor=south]{};
    \filldraw[black] (6,-2) circle (1pt) node[anchor=south]{};
    \filldraw[black] (6.5,-2) circle (1pt) node[anchor=south]{};
    \filldraw[black] (7,-2) circle (1pt) node[anchor=south]{};
  
 \draw[-,color=red]        (1,-2.1)   -- (2,-2.1);  
 \draw[-,color=green!50!black,dashed]        (1.5,-2.2)   -- (3,-2.2);  
 \draw[-,color=green!50!black]        (3,-2.2)   -- (4.5,-2.2);  
 \draw[-,color=blue,dashed]        (5.5,-2.1)   -- (6.5,-2.1);
 \draw[-,color=blue]        (6.5,-2.1)   -- (7,-2.1);
 
      \node (D) at (-0.5,-3) {$\rho'=$};
 \draw[-]        (0,-3)   -- (10,-3);
 \draw[->]        (10,-3)   -- (11,-3);
    \filldraw[black] (1,-3) circle (1pt) node[anchor=south]{};
    \filldraw[black] (2,-3) circle (1pt) node[anchor=south]{};
    \filldraw[black] (3,-3) circle (1pt) node[anchor=south]{};
    \filldraw[black] (3.5,-3) circle (1pt) node[anchor=south]{};
    \filldraw[black] (4,-3) circle (1pt) node[anchor=south]{};
    \filldraw[black] (5,-3) circle (1pt) node[anchor=south]{};
    \filldraw[black] (6.5,-3) circle (1pt) node[anchor=south]{};
    \filldraw[black] (7.5,-3) circle (1pt) node[anchor=south]{};
    \filldraw[black] (8,-3) circle (1pt) node[anchor=south]{};
    \filldraw[black] (8.5,-3) circle (1pt) node[anchor=south]{};
    \filldraw[black] (9,-3) circle (1pt) node[anchor=south]{};
    \filldraw[black] (9.5,-3) circle (1pt) node[anchor=south]{};
  
 \draw[-,color=red,dashed]        (1,-3.1)   -- (3,-3.1);  
 \draw[-,color=red]        (3,-3.1)   -- (4,-3.1);  
 \draw[-,color=green!50!black,dashed]        (3.5,-3.2)   -- (5,-3.2);  
 \draw[-,color=green!50!black]        (5,-3.2)   -- (6.5,-3.2);  
 \draw[-,color=blue,dashed]        (7.5,-3.1)   -- (9,-3.1);
 \draw[-,color=blue]        (9,-3.1)   -- (9.5,-3.1);

\end{tikzpicture}
    \caption{Illustration of the merging of two multi-pumpings, as per the construction of Lem.~\ref{lem:mergepumping}\label{fig:mergepumping}}
\end{figure}

\begin{lemma}\label{lem:mergepumping}
  Given an LTS $\lts$ and a run $\rho\in\runs$, let $\rho_1$ and $\rho_2$ be two multi-pumpings of~$\rho$.
  Then, there exists a multi-pumping $\rho'$ of $\rho$ such that $\rho'$ is also a multi-pumping of $\rho_1$ and of $\rho_2$.
\end{lemma}

The idea behind the proof is to do all the pumpings of $\rho_1$ and $\rho_2$ in the right order in order to merge the
two pumpings.

\begin{proof}
  As $\rho_1$ and $\rho_2$ are two multi-pumping of $\rho$, there are $n_1$ and $n_2$ such that $\rho_1$ is a $n_1$-pumping of $\rho$
  and $\rho_2$ is a $n_2$-pumping of $\rho$.
  Let us prove the result by induction over $n_1+n_2$.
  \begin{itemize}
    \item $n_1+n_2 =0$ : In that case $n_1=n_2=0$ and $\rho_1=\rho_2=\rho$ are trivial pumpings. Therefore the result trivially holds.
    \item Assume now that $n_1+n_2\geq 1$.
    If $n_1=0$, then $\rho_1=\rho$ and the result trivially holds.
    The same is true if $n_2=0$.
    Now assume that $n_1>0$ and $n_2>0$.
    Then by definition $\exists i_1<j_1, \rho[i_1]=\rho[j_1]$ such that there exists a $(n_1-1)$-pumping $\Tilde{\rho}_1$ of $\rho\sub{i_1}{}$ and $\rho_1=\rho\sub{}{i_1}\rho\sub{i_1}{j_1}^{l_1-1}\Tilde{\rho}_1$, where $l_1>0$.
    In the same way, $\exists i_2<j_2, \rho[i_2]=\rho[j_2]$ such that there exists a $(n_2-1)$-pumping $\Tilde{\rho}_2$ of $\rho\sub{i_2}{}$ and $\rho_2=\rho\sub{}{i_2}\rho\sub{i_2}{j_2}^{l_2-1}\Tilde{\rho}_2$, where $l_2>0$.
    Without loss of generality assume that $i_1\leq i_2$.
    By construction, as $i_1\leq i_2$, $\rho_2\sub{i_1}{}$ is a $n_2$-pumping of $\rho\sub{i_1}{}$, and $\Tilde{\rho}_1$ is by definition a $(n_1-1)$-pumping of $\rho\sub{i_1}{}$.
    By induction hypothesis, there exists a multi-pumping $\Tilde{\rho}'$ of $\rho\sub{i_1}{}$ that is also a multi-pumping of $\rho_2\sub{i_1}{}$ and $\Tilde{\rho}_1$.
    Then, by construction, the run $\rho'=\rho\sub{}{i_1}\rho\sub{i_1}{j_1}^{l_1-1}\Tilde{\rho}'$ is a multi-pumping of both $\rho_1$ and $\rho_2$.\qedhere
  \end{itemize}
\end{proof}

Note that the merging $\rho'$ is not unique.
For example, if $i_1=i_2$, one can do the pumping in any order and the resulting pumping will work, inducing two different multi-pumping.
In the following, when we refer to the merging of two multi-pumping, we fix an arbitrary one.

  \begin{figure}[h!]
  \centering
    \begin{tikzpicture}[thick]
    \node (A) at (-0.5,-0.05) {$\rho=$};
 \draw[-]        (0,0)   -- (3.5,0);
 \draw[->]        (3.5,0)   -- (4.5,0);
    \filldraw[black] (0,0) circle (1pt) node[anchor=south]{$a$};
    \filldraw[black] (0.5,0) circle (1pt) node[anchor=south]{$b$};
    \filldraw[black] (1,0) circle (1pt) node[anchor=south]{$c$};
    \filldraw[black] (1.5,0) circle (1pt) node[anchor=south]{$d$};
    \filldraw[black] (2,0) circle (1pt) node[anchor=south]{$e$};
    \filldraw[black] (2.5,0) circle (1pt) node[anchor=south]{$b$};
    \filldraw[black] (3,0) circle (1pt) node[anchor=south]{$d$};
    
 \draw[-,color=blue]        (0.5,-0.1)   -- (2.5,-0.1);

       \node (A) at (-0.5,-1) {$\rho'=$};
 \draw[-]        (0,-1)   -- (5.5,-1);
 \draw[->,dashed]        (5.5,-1)   -- (6.5,-1);
    \filldraw[black] (0,-1) circle (1pt) node[anchor=south]{$a$};
    \filldraw[black] (0.5,-1) circle (1pt) node[anchor=south]{$b$};
    \filldraw[black] (1,-1) circle (1pt) node[anchor=south]{$c$};
    \filldraw[black] (1.5,-1) circle (1pt) node[anchor=south]{$d$};
    \filldraw[black] (2,-1) circle (1pt) node[anchor=south]{$e$};
    \filldraw[black] (2.5,-1) circle (1pt) node[anchor=south]{$b$};
    \filldraw[black] (3,-1) circle (1pt) node[anchor=south]{$c$};
    \filldraw[black] (3.5,-1) circle (1pt) node[anchor=south]{$d$};
    \filldraw[black] (4,-1) circle (1pt) node[anchor=south]{$e$};
    \filldraw[black] (4.5,-1) circle (1pt) node[anchor=south]{$b$};
    \filldraw[black] (5,-1) circle (1pt) node[anchor=south]{$d$};
    
 \draw[-,color=blue,dashed]        (0.5,-1.1)   -- (2.5,-1.1);
 \draw[-,color=blue]        (2.5,-1.1)   -- (4.5,-1.1);

 \draw[-,dotted]        (3,-0.1)   -- (3,-0.7);
 
\end{tikzpicture}
    \caption{Graphical representation of 
    Lem.~\ref{lem:pumpingsameprefix}\label{fig:pumpingsameprefix}}
\end{figure}

\begin{lemma}\label{lem:pumpingsameprefix}
  Given an LTS $\lts$, a run $\rho\in\runs$ and $\rho'$ a multi-pumping of $\rho$, 
  let $i_0$ be the first index of a repeating state in
  $\rho$, that is, there exists a $0\leq i'< i_0$ such that $\rho[i']=\rho[i_0]$ and for all $0\leq j<j' <i_0$,
  $\rho[j]\neq\rho[j']$. Then, $\rho'\sub{0}{i_0+1}=\rho\sub{0}{i_0+1}$.
\end{lemma}

The intuitive idea behind this lemma is that a pumping can not modify the run before the first loop.
After this loop, the run and its pumping can differ, if for example the pumping is pumping that
particular loop, but before that point the two are the same.
For example, look at \figurename~\ref{fig:pumpingsameprefix}.
The first loop pumped in the pumping starts at the second state of the run, and the first difference between
the original run and its pumping happens right after the first reoccurring state.
Let us prove that formally.

\begin{proof}[Proof of Lem.~\ref{lem:pumpingsameprefix}]
  First, note that a multi-pumping is an $n$-pumping for some fixed $n$.
  By construction of a $n$-pumping by a sequence of $n$ distinct pumpings, it is enough to show the result
  for only one pumping. Then, by a trivial induction, it will hold for any multi-pumping.
  
  Assume that $\rho'$ is a pumping of $\rho$.
  Then, by definition, there are two positions $i_1<i_2$ such that $\rho[i_1]=\rho[i_2]$.
  Then, $\rho\sub{i_1}{i_2}$ is a loop, and
  there exists a $l>0$ such that $\rho'=\rho\sub{}{i_1}\rho\sub{i_1}{i_2}^l\rho[i_2..]$.
  By definition of $i_0$, we have that $i_0\leq i_2$.
  If $i_0< i_1$, then trivially $\rho'\sub{0}{i_0}\rho'[i_0]=\rho\sub{0}{i_0}\rho[i_0]$.
  If $i_1\leq i_0\leq i_2$, then
  $\rho'\sub{0}{i_0}\rho[i_0] = \rho\sub{0}{i_1}\rho\sub{i_1}{i_0}\rho[i_0]=\rho\sub{0}{i_0}\rho[i_0]$.
\end{proof}

  \begin{figure}[H]
  \vspace{5mm}
  \centering
    \begin{tikzpicture}[thick]

 \draw[-]        (0,0)   -- (1.9,0);
  \draw[-]        (2,0)   -- (6,0);
 \draw[->]        (6,0)   -- (7,0);

    \node (A) at (1,0.15) {$\pi$};
    \filldraw[black] (3,0) circle (1pt) node[anchor=south]{};
    \filldraw[black] (4,0) circle (1pt) node[anchor=south]{};

 \draw[-,color=blue]        (3,-0.1)   -- (4,-0.1);

  \draw[-]        (0,-1)   -- (1.9,-1);
  \draw[-]        (2,-1)   -- (6,-1);
 \draw[->]        (6,-1)   -- (7,-1);

    \node (A) at (1,-0.85) {$\pi$};
    \filldraw[black] (3,-1) circle (1pt) node[anchor=south]{};
    \filldraw[black] (4,-1) circle (1pt) node[anchor=south]{};
    \filldraw[black] (5,-1) circle (1pt) node[anchor=south]{};

 \draw[-,color=blue,dashed]        (3,-1.1)   -- (4,-1.1);  
 \draw[-,color=blue]        (4,-1.1)   -- (5,-1.1);  

\end{tikzpicture}
    \caption{Graphical representation of Lem.~\ref{lem:prefixpumping}\label{fig:prefixpumping}}
\end{figure}

\begin{lemma}\label{lem:prefixpumping}
  Given an LTS $\lts$, a finite path $\pi$, and two runs $\pi\rho\in\runs$, and $\pi\rho'\in\runs$
  such that $\rho'$ is a multi-pumping of $\rho$, then $\pi\rho'$ is a multi-pumping of $\pi\rho$.
  
\end{lemma}

The idea of this lemma is that adding a finite prefix from both a run and a multi-pumping
of this run will result in a new run and a multi-pumping of it. Intuitively, it makes sense as
adding a finite prefix on both does not change the loops of $\rho$ and how they are pumped in
$\rho'$, as illustrated by \figurename~\ref{fig:prefixpumping}. This is proven formally as follows.

\begin{proof}[Proof of Lem.~\ref{lem:prefixpumping}]
  First, note that a multi-pumping is a $n$-pumping for some finite $n$.
  By construction of a $n$-pumping by a sequence of $n$ distinct pumpings, it is enough to show the result
  for only one pumping. Then, by a trivial induction, it will hold for any multi-pumping.
  
  Assume that $\rho'$ is a pumping of $\rho$.
  Then there are two positions $i_1<i_2$ such that $\rho[i_1]=\rho[i_2]$ and $l>0$ such
  that $\rho'=\rho\sub{0}{i_1}\rho\sub{i_1}{i_2}^l\rho\sub{i_2}{}$.
  Then, $\pi\rho'=\pi\rho\sub{0}{i_1}\rho\sub{i_1}{i_2}^l\rho\sub{i_2}{}$, and by definition, $\pi\rho'$
  is a pumping of $\pi\rho$.
\end{proof}

\begin{lemma}\label{lem:pumpingdeviation}
  Given an lts $\lts$, a run $\rho\in\runs$, and a pumping $\rho'$ of $\rho$ such that $\rho'=\rho\sub{0}{i_1}\rho\sub{i_1}{i_2}^l\rho\sub{i_2}{}$, then
  one of the following holds for any position $j\geq 0$:
  \begin{itemize}
      \item $\rho'\sub{j}{}$ is a pumping of $\rho\sub{j}{}$.
      \item There is $i_1\leq j'\leq j$ such that $\rho'\sub{j}{}=\rho\sub{j'}{}$.
      \item There is a cycle $\pi$ and $i_1\leq j'\leq j$ such that $\rho'\sub{j}{}=\pi\rho\sub{j'}{}$.
  \end{itemize}
\end{lemma}

The idea of this lemma is that every position of a pumping can be matched to a position
of the original run. Depending on whether the position is before the added loop, in the middle of it or
afterwards, the structure is a bit different, but the matching still holds.
Mainly, this lemma is a tool used in other proofs.
As the proof is rather hermetic, we will explain the idea of this lemma with
some figures.

 \begin{figure}[H]
  \centering
  
\begin{tikzpicture}[thick]
    \node (A) at (-0.5,-0.05) {$\rho=$};
  \draw[->]        (0,0)   -- (10,0);
    \filldraw[black] (2,0) circle (1pt) node[anchor=south]{$i_1$};
    \filldraw[black] (4,0) circle (1pt) node[anchor=south]{$i_2$};
    
  \draw[-,color=blue]        (2,-0.1)   -- (4,-0.1);
  
    \node (B) at (-0.5,-2) {$\rho'=$};
  \draw[->]        (0,-2)   -- (10,-2);
    \filldraw[black] (2,-2) circle (1pt) node[anchor=south]{$i_1$};
    \filldraw[black] (4,-2) circle (1pt) node[anchor=south]{$i_2$};
    \filldraw[black] (6,-2) circle (1pt) node[anchor=south]{$ $};
    \filldraw[black] (8,-2) circle (1pt) node[anchor=south]{$ $};
    
  \draw[-,color=blue,dashed]        (2,-2.1)   -- (6,-2.1);
  \draw[-,color=blue]        (6,-2.1)   -- (8,-2.1);
    
    \filldraw[black] (1,-2) circle (1pt) node[anchor=south]{$j$};
    \filldraw[black] (1,0) circle (1pt) node[anchor=south]{$j$};
    \draw[->]        (1,-0.7)   -- (10,-0.7);
    \draw[->]        (1,-2.7)   -- (10,-2.7);
    \draw[dotted,-]  (1,0)   -- (1,-0.7);
    \draw[dotted,-]  (1,-2)   -- (1,-2.7);
    \node (C) at (0.4,-0.7) {$\rho[j..]=$};
    \node (D) at (0.37,-2.7) {$\rho'[j..]=$};
    
\end{tikzpicture}
  \caption{\label{fig:pumping1}First case : $j<i_1$.}
\end{figure}

Consider the run $\rho$ and the position $j$ describe in \figurename~\ref{fig:pumping1},
and the pumping $\rho'$ of $\rho$ where two loop $i_1\cdots i_2$ are added.
Clearly, $\rho\sub{}{j}=\rho'\sub{}{j}$ and thus $\rho'\sub{j}{}$ is a pumping of $\rho\sub{j}{}$.

 \begin{figure}[H]
  \centering
  \begin{tikzpicture}[thick]
    \node (A) at (-0.5,-0.05) {$\rho=$};
  \draw[->]        (0,0)   -- (10,0);
    \filldraw[black] (2,0) circle (1pt) node[anchor=south]{$i_1$};
    \filldraw[black] (4,0) circle (1pt) node[anchor=south]{$i_2$};
  
    \draw[-,color=blue]        (2,-0.1)   -- (4,-0.1);
  
    \node (B) at (-0.5,-2) {$\rho'=$};
  \draw[->]        (0,-2)   -- (10,-2);
    \filldraw[black] (2,-2) circle (1pt) node[anchor=south]{$i_1$};
    \filldraw[black] (4,-2) circle (1pt) node[anchor=south]{$i_2$};
    \filldraw[black] (6,-2) circle (1pt) node[anchor=south]{};
    \filldraw[black] (8,-2) circle (1pt) node[anchor=south]{};
    
  \draw[-,color=blue,dashed]        (2,-2.1)   -- (6,-2.1);
  \draw[-,color=blue]        (6,-2.1)   -- (8,-2.1);
    
    \filldraw[black] (9,-2) circle (1pt) node[anchor=south]{$j$};
    \filldraw[black] (5,0) circle (1pt) node[anchor=south]{$j'$};
    \draw[->]        (5,-0.7)   -- (10,-0.7);
    \draw[->]        (9,-2.7)   -- (10,-2.7);
    \draw[dotted,-]  (5,0)   -- (5,-0.7);
    \draw[dotted,-]  (9,-2)   -- (9,-2.7);
    \node (C) at (4.37,-0.7) {$\rho[j'..]=$};
    \node (D) at (8.37,-2.7) {$\rho'[j..]=$};
    
\end{tikzpicture}
  \caption{\label{fig:pumping2}Second case : $j>i_2+(i_2-i_1)(l-1)$.}
\end{figure}

Now, consider the same run and same pumping, but consider that $j>i_2+(i_2-i_1)(l-1)$, as represented
in \figurename~\ref{fig:pumping2}. Here, two loops are added, therefore $l=2$.
As $j$ is after every added loops in $\rho'$, one can find a $j'$ such that
$\rho'\sub{j}{}=\rho\sub{j'}{}$. Note that if $j$ is in the last loop $i_1\cdots i_2$, then
$j>i_2+(i_2-i_1)(l-1)$ and this idea still works.

 \begin{figure}[H]
 \vspace{-5mm}
  \centering
  \begin{tikzpicture}[thick]
    \node (A) at (-0.5,-0.05) {$\rho=$};
  \draw[->]        (0,0)   -- (10,0);
    \filldraw[black] (2,0) circle (1pt) node[anchor=south]{$i_1$};
    \filldraw[black] (4,0) circle (1pt) node[anchor=south]{$i_2$};
  
    \draw[-,color=blue]        (2,-0.1)   -- (4,-0.1);
  
    \node (B) at (-0.5,-2) {$\rho'=$};
  \draw[->]        (0,-2)   -- (10,-2);
    \filldraw[black] (2,-2) circle (1pt) node[anchor=south]{$i_1$};
    \filldraw[black] (4,-2) circle (1pt) node[anchor=south]{$i_2$};
    \filldraw[black] (6,-2) circle (1pt) node[anchor=south]{};
    \filldraw[black] (8,-2) circle (1pt) node[anchor=south]{};
    
  \draw[-,color=blue,dashed]        (2,-2.1)   -- (6,-2.1);
  \draw[-,color=blue]        (6,-2.1)   -- (8,-2.1);
    
    \filldraw[black] (5,-2) circle (1pt) node[anchor=south]{$j$};
    \filldraw[black] (7,-2) circle (1pt) node[anchor=south]{$\Tilde{j}$};
    \filldraw[black] (3,0) circle (1pt) node[anchor=south]{$j'$};
    \draw[->]        (3,-0.7)   -- (10,-0.7);
    \draw[->]        (5,-2.7)   -- (10,-2.7);
    \draw[dotted,-]  (3,0)   -- (3,-0.7);
    \draw[dotted,-]  (5,-2)   -- (5,-2.7);
    \draw[dotted,-]  (7,-2)   -- (7,-2.7);
    \node (C) at (2.4,-0.7) {$\rho[j'..]=$};
    \node (D) at (4.37,-2.7) {$\rho'[j..]=$};
    \node (E) at (6,-3) {$\pi$};
    \node (F) at (8.5,-3) {$\rho'[\Tilde{j}..]=\rho[j'..]$};
    
\end{tikzpicture}
  \caption{\label{fig:pumping3}Third case : $i_1\leq j\leq i_2+(i_2-i_1)(l-1)$}
\end{figure}

The last case is the most technical. Consider the same run $\rho$ and same pumping $\rho'$, but
with $i_1\leq j\leq i_2+(i_2-i_1)(l-1)$, that is, $j$ is somewhere inside an added loop,
as described in \figurename~\ref{fig:pumping3}.
Then, $\rho'\sub{j}{}$ may not be a pumping of any $\rho\sub{j'}{}$, as $\rho'\sub{j}{}$ starts with a bit
of the loop, then has some complete occurrences of the loop, and then continues in the same way
as $\rho$. Yet, one can consider the last occurrence of $\rho'[j]$, that is, the same vertex but
in the last loop, let say $\rho'[\Tilde{j}]$. Then, $\rho'\sub{j}{}$ can be described as a
loop $\pi=\rho'\sub{j}{\Tilde{j}}$ followed by $\rho'\sub{\Tilde{j}}{}$. Then, by the same argument as
in the second case, there exists a $j'$ such that $\rho\sub{j'}{}=\rho'\sub{\Tilde{j}}{}$.

Those ideas are formalised in the following proof.

\begin{proof}[Proof of Lem.~\ref{lem:pumpingdeviation}]
  Let us look at every possibility.
  
  Firstly, assume that $j< i_1$. Then, $\rho'\sub{j}{}=\rho\sub{j}{i_1}\rho\sub{i_1}{i_2}^l\rho\sub{i_2}{}$.
  By definition, $\rho\sub{j}{}$ is a pumping of $\rho\sub{j}{}$.
  
  Secondly, assume that $j> i_2+(i_2-i_1)(l-1)$. Then, let $j'=j-(i_2-i_1)(l-1)$.
  We have $j'>i_2$, and thus by construction $\rho'\sub{j}{}=\rho\sub{j'}{}$.
  
  Thirdly, we have $i_1\leq j\leq i_2+(i_2-i_1)(l-1)$. There is a $0\leq l'\leq l-2$ such that
  \[i_1+(i_2-i_1)l'\leq j\leq i_1+(i_2-i_1)(l'+1).\]
  Let $j'=j-(i_2-i_1)l'$, and let us write
  \[\rho'=\rho\sub{0}{i_1}\rho\sub{i_1}{i_2}^{l'}\rho\sub{i_1}{j'}\rho\sub{j'}{i_2}\rho\sub{i_1}{i_2}^{l-l'-2}\rho\sub{i_1}{j'}\rho\sub{j'}{i_2}\rho\sub{i_2}{}\,.\]
  Note that
  $\rho'\sub{j}{} =\rho\sub{j'}{i_2}\rho\sub{i_1}{i_2}^{l-l'-2}\rho\sub{i_1}{j'}\rho\sub{j'}{i_2}\rho\sub{i_2}{}$.
  
\noindent Let $\pi=\rho\sub{j'}{i_2}\rho\sub{i_1}{i_2}^{l-l'-2}\rho\sub{i_1}{j'}$.
  Note that $\pi[0]=\rho[j']\in\successor(\pi[|\pi|-1])$ and $\pi$ is a cycle.
  Therefore, $\rho'[j..]=\pi\rho\sub{j'}{}$.
\end{proof}

Now, we show some properties of pumpings in regard to the satisfaction of formulas.
Those properties are necessary in order to use pumpings and still keep satisfaction, or invalidation,
of the formula.

  \begin{figure}[H]
  \centering
    \begin{tikzpicture}[thick]

 \draw[-]        (0,0)   -- (1.9,0);
  \draw[-]        (2,0)   -- (6,0);
 \draw[->]        (6,0)   -- (7,0);

    \node (A) at (1,0.15) {$\pi$};
    \filldraw[black] (3,0) circle (1pt) node[anchor=south]{};
    \filldraw[black] (5,0) circle (1pt) node[anchor=south]{};

 \draw[<->,dashed]        (3,0.1)   -- (5,0.1);
    \node (B) at (4,0.3) {$k$};
    
    \draw [decorate,decoration={brace,amplitude=5pt,mirror,raise=1ex}]
  (3,0) -- (5,0) node[midway,yshift=-2em]{$\not\models\psi$};

 \draw[-,dotted]        (2,0.1)   -- (2,-1.1);
    \node (C) at (2.8,-1.1) {$\not\models\fpinf\psi$};
    
     \draw[-,dotted]        (0,0.1)   -- (0,-1.6);
    \node (D) at (0.8,-1.6) {$\not\models\fpinf\psi$};

\end{tikzpicture}
    \caption{Visual representation of Lem.~\ref{lem:pumpingloopstart}\label{fig:pumpingloopstart}}
\end{figure}

\begin{lemma}\label{lem:pumpingloopstart}
  Given an LTS $\lts$, a finite path $\pi$, a run $\rho=\pi\rho'\in\runs$, a bound $k\geq 0$, and
  a formula $\varphi\in\fpLTL$,
  if $\rho'[0]=\pi[0]$ and $(\rho',k)\not\models\varphi$, then $(\rho,k)\not\models\varphi$.
\end{lemma}

Note that here, $\pi$ is a loop.
In this lemma, the notion of pumping is hidden.
Mostly, when this lemma is used, the finite path $\pi$ is some loop in a pumping.
Intuitively, the idea is that adding a finite prefix to a run will not remove any faulty window,
just postpone them, as can be seen in \figurename~\ref{fig:pumpingloopstart}. For state formulas, we have to check that the initial state stays the same.
Over all, this result is not surprising as prompt Muller formulas can be seen as safety conditions.

\begin{proof}[Proof of Lem.~\ref{lem:pumpingloopstart}]
    We show the result by structural induction over the formula $\varphi$.
    \begin{itemize}
        \item If $\varphi=\alpha$ or $\varphi=\lnot\alpha$, then $(\rho',k)\not\models\varphi$
        is equivalent to $\rho'[0]\not\models\varphi$.
        As $\pi[0]=\rho'[0]$, $(\rho,k)\not\models\varphi$.
        \item If $\varphi=\psi_1\lor\psi_2$, then, $(\rho',k)\not\models\varphi$ is equivalent to
        $(\rho',k)\not\models\psi_1$ and $(\rho',k)\not\models\psi_2$.
        By induction hypothesis, this implies that $(\rho,k)\not\models\psi_1$ and
        $(\rho,k)\not\models\psi_2$, and therefore $(\rho,k)\not\models\varphi$.
        \item If $\varphi=\psi_1\land\psi_2$, then, $(\rho',k)\not\models\varphi$ is equivalent to
        $(\rho',k)\not\models\psi_1$ or $(\rho',k)\not\models\psi_2$. Without loss of generality,
        assume that $(\rho',k)\not\models\psi_1$.
        By induction hypothesis, this implies that $(\rho,k)\not\models\psi_1$.
        Therefore, $(\rho,k)\not\models\varphi$.
        \item If $\varphi=\fpinf\psi$, then $(\rho',k)\not\models\varphi$ implies that there exists a
        position $i$ such that for all $i\leq j\leq i+k$, $(\rho'\sub{j}{},k)\not\models\psi$.
        Therefore, $(\rho\sub{(|\pi|+j)}{},k)\not\models\psi$.
        This holds for each $i\leq j\leq i+k$, and therefore $(\rho,k)\not\models\varphi$.\qedhere
    \end{itemize}
\end{proof}

  \begin{figure}[H]
  \centering
    \begin{tikzpicture}[thick]
    
 \draw[-]        (0,0)   -- (1.9,0);
  \draw[-]        (2,0)   -- (7,0);
 \draw[->]        (7,0)   -- (8,0);

    \node (A) at (1,0.15) {$\pi$};

    \draw [decorate,decoration={brace,amplitude=5pt,mirror,raise=1ex}]
  (0,0) -- (1.9,0) node[midway,yshift=-2em]{$\not\models\psi$};

   \draw[-]        (0,-2)   -- (5.9,-2);
  \draw[-]        (6,-2)   -- (7,-2);
 \draw[->]        (7,-2)   -- (8,-2);

    \filldraw[black] (2,-2) circle (1pt) node[anchor=south]{};
    \filldraw[black] (4,-2) circle (1pt) node[anchor=south]{};
  
    \node (B) at (1,-1.85) {$\pi$};
    \node (C) at (3,-1.85) {$\pi$};
    \node (D) at (5,-1.85) {$\pi$};
    
        \draw [decorate,decoration={brace,amplitude=5pt,mirror,raise=1ex}]
  (0,-2) -- (5.9,-2) node[midway,yshift=-2em]{$\not\models\psi$};

\end{tikzpicture}
    \caption{Visual representation of Lem.~\ref{lem:pumpingloopstartfull}\label{fig:pumpingloopstartfull}}
\end{figure}

\begin{lemma}\label{lem:pumpingloopstartfull}
Given an LTS $\lts$, a finite path $\pi$, a run $\rho=\pi\rho'\in\runs$ such that $\pi[0]=\rho'[0]$,
a bound $k\geq 0$, a formula $\varphi$, and the run $\rho_l=\pi^l\rho'$ for $l>0$,
  if for all $0\leq j\leq |\pi|$, $(\rho[j..],k)\not\models\varphi$, then for all $0\leq j'\leq l\times|\pi|$, $(\rho_l[j..],k)\not\models\varphi$.
\end{lemma}

This lemma is quite technical and used only as a tool in other proofs.
The idea is that if a loop does not satisfy a formula, no iteration of that loop can satisfy the formula, as
represented in \figurename~\ref{fig:pumpingloopstartfull}.
It is a direct corollary of Lem.~\ref{lem:pumpingloopstart}, as formally shown in the following proof.

\begin{proof}[Proof of Lem.~\ref{lem:pumpingloopstartfull}]
  First, notice that $\pi[0]=\rho'[0]$ ensures that $\pi^l\rho'$ is indeed a valid run in the system.
  Now, consider $0\leq j'\leq l\times|\pi|$, and $j=j'$ modulo $|\pi|$.
  As $\rho_l\sub{((l-1)\times|\pi|)}{}=\rho$ by definition of $\rho_l$, we have that
  $\rho_l\sub{((l-1)\times|\pi|+j)}{}=\rho[j..]$.
  Moreover, by definition, $0\leq j\leq |\pi|$, and $\rho_l[j']=\rho[j]$.
  By noting that $\rho_l\sub{j'}{}=\rho_l\sub{j'}{((l-1)\times|\pi|+j)}\rho\sub{j}{}$ and
  $(\rho\sub{j}{},k)\not\models\varphi$, Lem.~\ref{lem:pumpingloopstart} is enough to conclude.
\end{proof}

We have now presented every tool needed to prove Prop.~\ref{prop:pumpingnobetter}

\begin{proof}[Proof of Prop.~\ref{prop:pumpingnobetter}]

  By definition, a multi-pumping is a $n$-pumping for some $n\geq 0$, that is, a multi-pumping is a succession of
  a finite number of pumpings.
  Therefore, to show the result for $\rho'$ a pumping of $\rho$ is enough to conclude with a structural induction.
  
  Let us prove by structural induction over $\varphi$ that if $\rho'$ is a pumping of $\rho$, then
  $(\rho,k)\not\models\varphi$ implies that $(\rho',k)\not\models\varphi$.
  \begin{itemize}
    \item $\varphi=\alpha$ or $\varphi=\lnot\alpha$ : Then the result trivially holds as $\rho[0]=\rho'[0]$ by construction.
    \item $\varphi=\psi_1\lor\psi_2$ : By definition, $(\rho,k)\not\models\varphi$ if $(\rho,k)\not\models\psi_1$ and $(\rho,k)\not\models\psi_2$.
    By induction hypothesis, this implies that $(\rho',k)\not\models\psi_1$ and $(\rho',k)\not\models\psi_2$, and therefore $(\rho',k)\not\models\varphi$.
    \item $\varphi=\psi_1\land\psi_2$ : By definition, $(\rho,k)\not\models\varphi$ if $(\rho,k)\not\models\psi_1$ or $(\rho,k)\not\models\psi_2$.
    Without loss of generality assume that $(\rho,k)\not\models\psi_1$.
    By induction hypothesis, this implies that $(\rho',k)\not\models\psi_1$ and therefore $(\rho',k)\not\models\varphi$.
    \item $\varphi=\fpinf\psi$ :
    By definition, there exists a position $i\geq 0$ such that for all $i\leq j\leq i+k$,
    $\rho\sub{j}{}\not\models\psi$.
    Write $\rho'=\rho\sub{0}{i_1}\rho\sub{i_1}{i_2}^l\rho\sub{i_2}{}$ with $l>0$ and $i_1<i_2$.
    Firstly, if $i_1<i$ then for all $ i\leq j \leq i+k$,
    $\rho'\sub{(j+(i_2 - i_1)(l-1))}{}=\rho\sub{j}{}$.
    Therefore for all $ i+(i_2 - i_1)(l-1)\leq j\leq i+(i_2 - i_1)(l-1)+k $,
    $\rho'\sub{j}{}\not\models\psi$.
    Secondly, we have $i\leq i_1$.
    Let us show that $\forall i\leq j\leq i+k$, $\rho'\sub{j}{}\not\models\psi$.
    By Lem.~\ref{lem:pumpingdeviation}, there are three cases.
    \begin{itemize}
      \item If $\rho'\sub{j}{}$ is a pumping of $\rho\sub{j}{}$, then by induction hypothesis $\rho'\sub{j}{}\not\models \psi$.
      \item If there exists $i_1\leq j'\leq j$ such that $\rho'\sub{j}{}=\rho\sub{j'}{}$ then,
      as $i\leq i_1$, we have that $i\leq j'\leq i+k$, so that $\rho\sub{j'}{}\not\models \psi$, and thus $\rho'\sub{j}{}\not\models \psi$.
      \item If there exists a cycle $\pi$ and $i_1\leq j'\leq j$ such that $\rho'\sub{j}{}=\pi\rho\sub{j'}{}$, then once again $i\leq j'\leq i+k$, so that $\rho\sub{j'}{}\not\models \psi$, and by Lem.~\ref{lem:pumpingloopstart} $\rho'\sub{j}{}\not\models \psi$.\qedhere
  \end{itemize}
  \end{itemize}
\end{proof}

\pboundgeneral*


\begin{proof}
  By structural induction over $\varphi$.
  \begin{itemize}
    \item $\varphi=\alpha$ or $\varphi=\neg\alpha$: The result trivially holds
    as the bound is irrelevant to the satisfaction of $\varphi$, and $\rho_N$ is a multi-pumping of itself.
    \item $\varphi=\psi_1\lor\psi_2$ : By definition, $(\rho_N,N)\not\models\varphi$
    if $(\rho_N,N)\not\models\psi_1$ and $(\rho_N,N)\not\models\psi_2$.
    By induction hypothesis, for all $k\geq N$, there are two multi-pumpings $\rho_k^1$ and $\rho_k^2$ of $\rho_N$ such that
    $(\rho_k^1,k)\not\models\psi_1$ and $(\rho_k^2,k)\not\models\psi_2$.
    By lemma \ref{lem:mergepumping}, there exists a multi-pumping $\rho_k$ of $\rho_N$ such that $\rho_k$ is also a
    multi-pumping of $\rho_k^1$ and $\rho_k^2$.
    Then, by lemma \ref{prop:pumpingnobetter}, $(\rho_k,k)\not\models\psi_1$ and $(\rho_k,k)\not\models\psi_2$.
    Therefore, $(\rho_k,k)\not\models\psi_1\lor\psi_2$.
    \item $\varphi=\psi_1\land\psi_2$ : By definition, $(\rho_N,N)\not\models\varphi$
    if $(\rho_N,N)\not\models\psi_1$ or $(\rho_N,N)\not\models\psi_2$.
    Without loss of generality, suppose that $(\rho_N,N)\not\models\psi_1$.
    By induction hypothesis, for all $k\geq N$, there exists a multi-pumping $\rho_k$ of $\rho_N$ such that
    $(\rho_k,k)\not\models\psi_1$.
    Then, by definition, $(\rho_k,k)\not\models\psi_1\land\psi_2$.
    \item $\varphi=\fpinf\psi$ : By definition, if $(\rho_N,N)\not\models\varphi$ then
    $\exists i, \forall 0\leq j \leq N, (\rho_N\sub{(i+j)}{},N)\not\models\psi$.
    As $N>|\states|$, there are two positions $i_1, i_2$ with $i\leq i_1<i_2\leq i+N$ such that $\rho_N[i_1]=\rho_N[i_2]$.
    Moreover, one can assume w.l.o.g. that $i_1$ and $i_2$ are chosen such that the finite path $\rho_N\sub{i_1}{i_2}$ is such that for each $i_1\leq j_1< j_2\leq i_2-1$, $\rho_N[j_1]\neq\rho_N[j_2]$.
    For each $i_1\leq j\leq i_2$, we have $(\rho_N\sub{j}{},N)\not\models\psi$.
    For each $i_1\leq j\leq i_2$, we can apply the induction hypothesis to obtain 
    a multi-pumping $\rho_k^j$ of $\rho_N\sub{j}{}$ such that $(\rho_k^j,k)\not\models\psi$.
    By Lem.~\ref{lem:prefixpumping}, for each $i_1\leq j\leq i_2$, we have that
    $\rho_N\sub{i_1}{j}\rho_k^j$ is a multi-pumping of $\rho_N[i_1..]$.
    By applying Lem.~\ref{lem:mergepumping} multiple times, there exists a run $\Tilde{\rho}_k$ such that for
    each $i_1\leq j\leq i_2$, $\Tilde{\rho}_k$ is a multi-pumping of $\rho_N\sub{i_1}{j}\rho_k^j$.
    Moreover, observe that it is not possible to have a loop between $i_1$ and $i_2$ in $\rho$. Therefore, by construction,
    for each $i_1\leq j\leq i_2$, we have that $\Tilde{\rho}_k\sub{(j-i_1)}{}$ is a
    multi-pumping of $\rho_k^j$ and therefore, by Prop.~\ref{prop:pumpingnobetter},
    $(\Tilde{\rho}_k\sub{(j-i_1)}{},k)\not\models\psi$.
    Consider the run $\Tilde{\rho}'_k=\Tilde{\rho}_k\sub{0}{(i_2-i_1)}^k\Tilde{\rho}_k\sub{(i_2-i_1)}{}$.
    By Lem.~\ref{lem:pumpingloopstartfull}, for all $0\leq j\leq k\times (i_2-i_1)$,
    $(\Tilde{\rho}'_k\sub{j}{},k)\not\models\psi$.
    Consider the run $\rho_k=\rho_N\sub{0}{i_1}]\Tilde{\rho}_k$.
    Then, as $i_2>i_1$, $k\times (i_2-i_1)>k$, for each $0\leq j\leq k$,
    $\rho_k\sub{(i_1+j)}{}=\Tilde{\rho}'_k\sub{j}{}$, and therefore $(\rho_k\sub{(i_1+j)}{},k)\not\models\psi$.
    By definition, $(\rho_k,k)\not\models\fpinf\psi$.\qedhere
  \end{itemize}
\end{proof}

\section{Small witness property}\label{app:smallwitness}

\smallwitness*

The proof of this proposition follows from the following lemmas.

\begin{lemma}\label{lm:witnesscorrectness}
Given an LTS $\lts$, a formula $\varphi\in\fpLTL$, a run $\rho\in\runs$ and a bound $k\geq 0$, if 
 $\witness$ is a witness
for $\rho$ with bound $k$, then for each $\psi\in\subform(\varphi)$, for each $i\in\witness(\psi)$,
$(\rho\sub{i}{},k)\not\models\psi$.
\end{lemma}

\begin{proof}
Let us prove this by structural induction over the formula $\varphi$.
\begin{itemize}
    \item If $\varphi=\alpha$, then for each $i\in\witness(\alpha)$, by definition, $\alpha\not\in\lbl(\rho[i])$.
    \item If $\varphi=\lnot\alpha$,
    then for each $i\in\witness(\lnot\alpha)$, by definition, $\alpha\in\lbl(\rho[i])$.
    \item If $\varphi=\psi_1\lor\psi_2$, consider a subformula $\psi$ of $\varphi$.
    Then assume that $\psi$ is a subformula of either $\psi_1$ or $\psi_2$, and by induction hypothesis,
    for each $i\in\witness(\psi)$, $(\rho\sub{i}{},k)\not\models\psi$.
    If $\psi$ is not a subformula of $\psi_1$ nor $\psi_2$, then $\psi=\varphi$.
    Let $i\in\witness(\psi)$.
    Since $i\in\witness(\psi_1)\cap\witness(\psi_2)$, then by induction
    hypothesis, $(\rho\sub{i}{},k)\not\models\psi_1$ and $(\rho\sub{i}{},k)\not\models\psi_2$.
    Thus, $(\rho\sub{i}{},k)\not\models\psi$.
    \item If $\varphi=\psi_1\land\psi_2$, consider a subformula $\psi$ of $\varphi$.
    Then assume that $\psi$ is a subformula of either $\psi_1$ or $\psi_2$, and by induction hypothesis,
    for each $i\in\witness(\psi)$, $(\rho\sub{i}{},k)\not\models\psi$.
    If $\psi$ is not a subformula of $\psi_1$ nor $\psi_2$, then $\psi=\varphi$.
    Let $i\in\witness(\psi)$.
    Since $i\in\witness(\psi_1)\cup\witness(\psi_2)$, 
    without loss of generality, assume $i\in\witness(\psi_1)$, then by induction
    hypothesis, $(\rho\sub{i}{},k)\not\models\psi_1$.
    Thus, $(\rho\sub{i}{},k)\not\models\psi$.
    \item If $\varphi=\fpinf\psi_1$, consider a subformula $\psi$ of $\varphi$.
    Then assume that $\psi$ is a subformula of $\psi_1$, by induction hypothesis,
    for each $i\in\witness(\psi)$, $(\rho\sub{i}{},k)\not\models\psi$.
    If $\psi$ is not a subformula of $\psi_1$, then $\psi=\varphi$.
    Let $i\in\witness(\psi)$.
    There is a $i'\geq i$ such that for all $i'\leq j\leq i'+k$,
    $j\in\witness(\psi_1)$.
    By induction hypothesis, for all $i'\leq j\leq i'+k$, $(\rho\sub{j}{},k)\not\models\psi_1$.
    Thus, $(\rho\sub{i}{},k)\not\models\psi$.\qedhere
\end{itemize}
\end{proof}

\begin{restatable}{lemma}{witnessconstruction}\label{lm:witnessconstruction}
Given an LTS $\ltsFull$, a formula $\varphi\in\fpLTL$, a run $\rho\in\runs$ 
and a bound $k\geq 0$, if $(\rho,k)\not\models\varphi$, then
there exists a finite path
  $\pi$ and a function $\witness$ such that for all run $\pi\rho'$, $\witness$ is a witness for $\pi\rho'$ with bound
  $k$ and $\max(\witness)< |\pi|$ and for each 
  $\psi\in\subform(\varphi)$, $|\witness(\psi)|\leq k+1$.    
\end{restatable}

\begin{proof}
Assume that $(\rho,k)\not\models\varphi$, and let us construct a witness function $\witness_{\rho,k}$,
such that for all $\psi\in\subform(\varphi)$, $|\witness_{\rho,k}(\psi)|\leq k+1$.
We will build the witness function top down over the syntax tree of $\varphi$.
We also
prove that this construction maintains an inductive invariant.

First, we define $\witness_{\rho,k}(\varphi)=\{0\}$. Note that for each $i\in\witness_{\rho,k}(\varphi)$,
that is $i=0$, we have $(\rho\sub{i}{},k)\not\models\varphi$.
Second, we assume that $\psi\in\subform(\varphi)$ is a subformula of $\varphi$ such that
$\witness_{\rho,k}(\psi)$ is defined,
$|\witness_{\rho,k}(\psi)|\leq k+1$
and for each $i\in\witness_{\rho,k}(\psi)$, $(\rho\sub{i}{},k)\not\models\psi$.
Then, let us explain how to define $\witness_{\rho,k}$ for the daughter subformulas of $\psi$ in
the syntax tree of $\varphi$.
Consider the form of $\psi$.
\begin{itemize}
    \item If $\psi$ is an atomic proposition $\alpha$ or $\lnot\alpha$, then we are done as $\psi$ has no
    non-trivial subformula.
    \item If $\psi=\psi_1\lor\psi_2$, then define
    $\witness_{\rho,k}(\psi_1)=\witness_{\rho,k}(\psi_2)=\witness_{\rho,k}(\psi)$
    Then, for each $i\in\witness_{\rho,k}(\psi_1)$, $i\in\witness_{\rho,k}(\psi)$.
    By induction hypothesis, $(\rho\sub{i}{},k)\not\models\psi$, and therefore, $(\rho\sub{i}{},k)\not\models\psi_1$.
    Similarly, $(\rho\sub{i}{},k)\not\models\psi_2$.
    Moreover, $\witness_{\rho,k}$ satisfies
    $\witness_{\rho,k}(\psi)=\witness_{\rho,k}(\psi_1)\cap\witness_{\rho,k}(\psi_2)$.
    Finally, as $|\witness_{\rho,k}(\psi)|\leq k+1$, $|\witness_{\rho,k}(\psi_1)|\leq k+1$ and
    $|\witness_{\rho,k}(\psi_2)|\leq k+1$.
    \item If $\psi=\psi_1\land\psi_2$, then define 
    $\witness_{\rho,k}(\psi_1)=\{i\in\witness_{\rho,k}(\psi)\mid (\rho[i..],k)\not\models\psi_1\}$
    and 
    $\witness_{\rho,k}(\psi_2)=\{i\in\witness_{\rho,k}(\psi)\mid (\rho\sub{i}{},k)\not\models\psi_2\}$.
    For each $i\in\witness_{\rho,k}(\psi)$, by induction hypothesis, $(\rho\sub{i}{},k)\not\models\psi$,
    and therefore, either $(\rho\sub{i}{},k)\not\models\psi_1$, and $i\in\witness_{\rho,k}(\psi_1)$
    or $(\rho\sub{i}{},k)\not\models\psi_2$, and $i\in\witness_{\rho,k}(\psi_2)$.
    We can conclude that $\witness_{\rho,k}(\psi)=\witness_{\rho,k}(\psi_1)\cup \witness_{\rho,k}(\psi_2)$,
    and that for each $i\in\witness_{\rho,k}(\psi_1)$ (resp. $\psi_2$),
    $(\rho\sub{i}{},k)\not\models\psi_1$ (resp. $(\rho\sub{i}{},k)\not\models\psi_2$).
    Finally, as $|\witness_{\rho,k}(\psi)|\leq k+1$, $|\witness_{\rho,k}(\psi_1)|\leq k+1$ and
    $|\witness_{\rho,k}(\psi_2)|\leq k+1$.
    \item If $\psi=\fpinf\psi_1$, then, by induction hypothesis, for each $i\in\witness_{\rho,k}(\psi)$,
    $(\rho\sub{i}{},k)\not\models\fpinf\psi_1$.
    Consider $i_{\max}=\max\{\witness_{\rho,k}(\psi)\}$.
    By definition of the semantics of the $\fpinf$ operator, there exists $i'_{\max} \geq i_{\max}$ such that 
    for each $i'_{\max}\leq j\leq i'_{\max} +k$, $(\rho\sub{j}{},k)\not\models\psi_1$.
    Then, define $\witness_{\rho,k}(\psi_1)=\{i'_{\max},i'_{\max} +1,\ldots,i'_{\max} +k\}$.
    By construction, for each $i\in\witness_{\rho,k}(\psi)$, $i\leq i_{\max}$ and therefore
    $i\leq i'_{\max}$. Thus, for each $i\in\witness_{\rho,k}(\psi)$, there exists a $i'\geq i$ such that
    for all $i'\leq j\leq i'+k$, $j\in\witness_{\rho,k}(\psi_1)$, by taking $i'=i'_{\max}$.
    Moreover, for each $i\in\witness_{\rho,k}(\psi_1)$, $(\rho\sub{i}{},k)\not\models\psi_1$ by construction.
    Finally, by construction $|\witness_{\rho,k}(\psi_1)|\leq k+1$.
\end{itemize}

Let $i=\max(\witness_{\rho,k})$ be the greatest index in the witness function.
Then, let $\pi=\rho[0..i]$.
Note that for any run $\pi\rho'$, $\witness_{\rho,k}$ is also a witness for $\pi\rho'$.
\end{proof}

\begin{lemma}\label{lm:witnessreduction}
Given an LTS $\ltsFull$, a formula $\varphi\in\fpLTL$ and a bound $k\geq 0$, assume that
there exists a finite path
  $\pi$ and a function $\witness$ such that for all run $\pi\rho'$, $\witness$ is a witness for $\pi\rho'$
  with bound $k$, $\max(\witness)< |\pi|$ and for each $\psi\in\subform(\varphi)$, $|\witness(\psi)|\leq k+1$.
  Then, there exists $\pi'$ and $\witness'$ such that $|\pi'|\leq (k+1)|\varphi|\,(|\states|+1)$, and for all run $\pi'\rho'$, $\witness'$ is a witness for $\pi'\rho'$
  with bound $k$ and $\max(\witness')\leq |\pi'|$.
\end{lemma}

\begin{proof}
Assume that there exists a finite path
  $\pi$ and a function $\witness$ such that for all run $\pi\rho'$, $\witness$ is a witness for $\pi\rho'$
  and $k$, $\max(\witness)< |\pi|$ and for each $\psi\in\subform(\varphi)$, $|\witness(\psi)|\leq k+1$,
  and that $|\pi|>(k+1)|\varphi|\,(|\states|+1)$.
  Let $I=\{i\mid \psi\in\subform(\varphi) \land i\in \witness(\psi)\}$ be the set of all
  positions that appear at some point in the witness function.
  As for each $\psi\in\subform(\varphi)$, $|\witness(\psi)|\leq k+1$, we have that $|I|\leq (k+1)|\varphi|$.
  Therefore, there are at least $|\states|+1$ consecutive positions in $\pi$ such that those positions
  do not belong to $I$.
  Thus, there are $0\leq i_1<i_2< |\pi|$ such that $\pi[i_1]=\pi[i_2]$, and $\pi[i_1..i_2]$ is
  a cycle and $\forall i_1\leq j\leq i_2,\pi[j]\not\in I$.
  Consider $\pi'=\pi\sub{0}{i_1}\pi\sub{i_2}{|\pi|}$.
  Let $\witness'$ be a function defined as follows. For every $\psi\in\subform(\varphi)$,
  $\witness'(\psi) = \{i\in\witness(\psi)\mid i< i_1\}\cup\{i-(i_2-i_1)\mid i\in\witness(\psi)\land i>i_2\}$.
  In particular, no position in $\{i_1,\ldots,i_2\}$ can belong to $\witness(\psi)$ by assumption.
  Note that $\max(\witness')=\max(\witness)-(i_2-i_1)$, so that $\max(\witness')\leq |\pi'|$, and for any run $\pi'\rho'$,
  for any $i_2< i\leq \max(\witness)$, $(\pi'\rho')[i-(i_2-i_1)]=\pi'[i-(i_2-i_1)]=\pi[i]$,
  and for any $0\leq i\leq i_1$, $(\pi'\rho')[i]=\pi'[i]=\pi[i]$.
  Let us prove that for any run $\pi'\rho'$, $\witness'$ is a witness for $\pi'\rho'$ and $k$ by checking that
  the six rules of Def.~\ref{def:witness} hold for each subformula $\psi$ of $\varphi$.
  \begin{itemize}
      \item If $\psi=\alpha$ is an atomic proposition, then for all $i\in\witness'(\alpha)$, two cases can occur.
      If $i\leq i_1$ then $i\in\witness(\alpha)$ and $\pi'[i]=\pi[i]$, thus $\alpha\not\in\lbl(\pi'[i])$.
      If $i>i_1$ then $i+(i_2-i_1)\in\witness(\alpha)$ and $\pi'[i]=\pi[i+(i_2-i_1)]$,
      thus $\alpha\not\in\lbl(\pi'[i])$.
      \item If $\psi=\lnot\alpha$ is an atomic proposition, then for all $i\in\witness'(\alpha)$,
      two cases can occur.
      If $i\leq i_1$ then $i\in\witness(\alpha)$ and $\pi'[i]=\pi[i]$, thus $\alpha\in\lbl(\pi'[i])$.
      If $i>i_1$ then $i+(i_2-i_1)\in\witness(\alpha)$ and $\pi'[i]=\pi[i+(i_2-i_1)]$, thus $\alpha\in\lbl(\pi'[i])$.
      \item If $\psi=\psi_1\lor\psi_2$, then for all $i\in\witness'(\psi)$, 
      two cases can occur.
      If $i\leq i_1$ then $i$ is in $\witness(\psi)$, $\witness(\psi_1)$ and $\witness(\psi_2)$.
      Therefore, $i\in\witness'(\psi_1)\cap\witness'(\psi_2)$.
      If $i>i_2$, then $i+(i_2-i_1)$ is in $\witness(\psi)$, $\witness(\psi_1)$ and $\witness(\psi_2)$.
      Therefore, $i\in\witness'(\psi_1)\cap\witness'(\psi_2)$.
      \item If $\psi=\psi_1\land\psi_2$, then for all $i\in\witness'(\psi)$,
      two cases can occur.
      If $i\leq i_1$ then $i$ is in $\witness(\psi)$ and thus is in $\witness(\psi_1)$ or $\witness(\psi_2)$.
      Therefore, $i\in\witness'(\psi_1)\cup\witness'(\psi_2)$.
      If $i>i_2$, then $i+(i_2-i_1)$ is in $\witness(\psi)$, and thus $\witness(\psi_1)$ or $\witness(\psi_2)$.
      Therefore, $i\in\witness'(\psi_1)\cup\witness'(\psi_2)$.
      \item If $\psi=\fpinf\psi_1$, then for all $i\in\witness'(\psi)$, let us show that there exists $i'$ such that
      for each $i'\leq j\leq i'+k$, $j\in\witness'(\psi_1)$.
      Two cases can occur for each $i\in \witness'(\psi)$.
      
      If $i\leq i_1$ then $i\in\witness(\psi)$, and there exists $i'\leq i$ such that for each $i'\leq j\leq i'+k$, $j\in\witness(\psi_1)$. Note that either $i'+k<i_1$ or $i'>i_2$ since $\{i',\ldots,i'+k\}\in I$. If $i'+k<i_1$, then for each $i'\leq j\leq i'+k$, $j\in\witness'(\psi_1)$, and $i'\geq i$. If $i'>i_2$, then for each $i'\leq j\leq i'+k$, $j-(i_2-i_1)\in\witness'(\psi_1)$. Finally, note that $i'-(i_2-i_1)\geq i$.
      
      If $i> i_2$ then $i+(i_2-i_1)\in\witness(\psi)$ and there exists $i'\leq i+(i_2-i_1)$ such that for each $i'\leq j\leq i'+k$, $j\in\witness(\psi_1)$.
      Then, as $i'>i_2$, we have for each $i'\leq j\leq i'+k$ that $j-(i_2-i_1)\in\witness'(\psi_1)$ and $i'-(i_2-i_1)\geq i$.
      \item As $0<i_1$, then $0\in\witness'(\varphi)$.
  \end{itemize}
    As long as $|\pi| > (k+1)|\varphi|\,(|\states|+1)$, then this procedure of shortening $\pi$ by
    removing a cycle and
    modifying $\witness$ can be applied to get a strictly shorter $\pi$.
    This procedure terminates and results in a $\pi$ such that $|\pi|\leq (k+1)|\varphi|\,(|\states|+1)$.
\end{proof}

This concludes the proof of Prop.~\ref{prop:smallwitness}, as
Lem.s~\ref{lm:witnessconstruction} and~\ref{lm:witnessreduction} detail how to get a small witness, and in the other direction
$(\rho,k)\not\models\varphi$ is implied by Lem.~\ref{lm:witnesscorrectness} and the fact that $0\in\witness_{\rho,k}(\varphi)$.

\section{Initialized Muller formulas} 
\label{app:secInitSys}
\subsection{Universal model checking}

\smallwitnessinitialised*

In order to prove the above proposition, we restate Lem.~\ref{lem:pumpingloopstart} in our new context.

\begin{lemma}\label{lem:pumpingloopstartcorollary}
  Given an LTS $\ltsFull$, a finite path $\pi$, a run $\rho=\pi\rho'\in\runs$, a bound $k\geq 0$, and
  a formula $\varphi\in\fpLTLplus$,
  if $(\rho',k)\not\models\varphi$, then $(\rho,k)\not\models\varphi$.
\end{lemma}

This lemma is a direct corollary of Lem.~\ref{lem:pumpingloopstart}
by noting that the condition upon the initial states is irrelevant when all
state formulas are in the scope of an $\fpinf$ operator.

\begin{proof}[Proof of Prop.~\ref{prop:smallwitnessinitialised}]
  Assume that there exists a run
  $\rho\in\runs$ such that $(\rho,k)\not\models\f\varphi$.
  By definition of the $\f$ operator, this implies that for all $i\geq 0$,
  $(\rho\sub{i}{},k)\not\models\varphi$.
  By Prop.~\ref{prop:smallwitness}, for each $i\geq 0$,
  there exists a finite path $\pi_i$ and a function $\witness_i$ such that for all
  run $\pi_i\rho'\in\cyl(\pi_i)$, $\witness_i$ is a witness for $\pi_i\rho'$ and $k$ and
  $|\pi_i|\leq (k+1)|\varphi|(|\states|+1)$.
  As this holds for infinitely many $i$, there are at least two $i_1<i_2$ such that
  $\pi_{i_1}=\pi_{i_2}$ and $\witness_{i_1}=\witness_{i_2}$, and $i_2-i_1\geq (k+1)|\varphi|(|\states|+1)$.
  Let us call $\pi_{i_1}=\pi$ and $\witness_{i_1}=\witness$.
  By definition of the witness, $\pi[0]=\rho[i_1]$, therefore, $\pi[0]$ is reachable from
  $\rho[0]$, as $\rho\sub{}{(i_1+1)}$ is a valid finite path.
  Moreover, as $i_2-i_1\geq|\pi|$, $\rho\sub{(i_1+|\pi|-1)}{(i_2+1)}$ is a valid finite path.
  As $\rho[(i_1+|\pi|-1)]=\pi[|\pi|-1]$ and $\rho[i_2]=\pi[0]$,
  $\pi[0]$ is reachable from $\pi[|\pi|-1]$.
  
  Conversely, assume that there exists a finite path
  $\pi$ and a function $\witness$ such that for any run $\pi\rho'$, $\witness$ is a witness for $\pi\rho'$ and $k$, $|\pi|\leq (k+1)|\varphi|(|\states|+1)$, $\max(\witness)< |\pi|$ and $\pi[0]$ is reachable
  from $\rho[0]$ and $\pi[|\pi|-1]$.
  As $\pi[0]$ is reachable from $\rho[0]$, let us define $\rho[0]w_0\pi[0]$ a finite path from $\rho[0]$
  to $\pi[0]$.
  As $\pi[0]$ is reachable from $\pi[|\pi|-1]$, let us define $\pi[|\pi|-1]w\pi[0]$ a finite path
  from $\pi[|\pi|-1]$ to $\pi[0]$.
  Let us define $\rho=\rho[0]w_0(\pi w)^\omega$.
  Let us prove that $(\rho,k)\not\models\f\varphi$.
  let us define, for each $l\geq 0$, $i_l=|\rho[0] w_0|+l(|\pi w|)$, such that for
  each $l\geq 0$, $\rho\sub{i_l}{}=(\pi w)^\omega$.
  By Prop.~\ref{prop:smallwitness}, for each $l\geq 0$, $\rho\sub{i_l}{}\not\models\varphi$.
  For each $i\geq 0$, there exists a  $l\geq 0$ such that $i_l\geq i$, and
  $\rho\sub{i}{}=\rho\sub{i}{i_l}\rho\sub{i_l}{}$.
  As $\varphi\in\fpLTLplus$, by Lem.~\ref{lem:pumpingloopstartcorollary}, one has $(\rho\sub{i}{},k)\not\models\varphi$.
  This holds for all $i\geq 0$, therefore, $(\rho,k)\not\models\f\varphi$.
\end{proof}

\TLTLnpcomp*

This result is obtained with the non-deterministic algorithm in Algorithm~\ref{alg:uniMCprexind}.

\begin{algorithm}
  \caption{$\conp$ algorithm for the universal model checking of
    $\TfpLTL$\label{alg:uniMCprexind}}
 \KwData{An LTS $\ltsFull$ and a formula $\varphi\in\TfpLTL$}
 \KwResult{whether $\lts\not\models \varphi$}
 \SetKwFunction{fun}{CheckW}
 \SetKwFunction{funpath}{CheckPath}
 \medskip
 \SetKwProg{myalg}{GuessAndCheck}{}{}
 \SetKwProg{check}{Check}{}{}
 \SetKwData{n}{N}
 \myalg{$(\lts, \varphi)$}{
   $\n \gets |\states|+1$\;
   guess a finite path $\pi$ such that $|\pi|\leq (N+1)|\varphi|(|\states|+1)$\;
   guess a function $\witness:\subform(\varphi)\mapsto 2^\mathbb{N}$ such that $\max(\witness)< |\pi|$\;
   \Return $(0\in\witness(\varphi))\land$\fun{$\lts, \witness, \varphi,\pi$}$\land$\funpath{$\lts, \pi$}\;
 }
 \SetKwProg{myFun}{}{}{}
 \myFun{\fun{$\lts, \witness, \varphi,\pi$}}{
   \If{$\varphi=\alpha$} {
     \Return $\forall i\in\witness(\varphi), a\not\in\lbl(\pi[i])$\;
   }
   \ElseIf{$\varphi=\lnot\alpha$} {
     \Return $\forall i\in\witness(\varphi), a\in\lbl(\pi[i])$\;
   }
   \ElseIf{$\varphi=\psi_1\lor\psi_2$}{
     \Return $(\witness(\varphi)=\witness(\psi_1)\cap\witness(\psi_2))\land$\fun{$\lts,\witness, \psi_{1},\pi$} $\land$ \fun{$\lts,\witness, \psi_{2},\pi$}\;
   }
   \ElseIf{$\varphi=\psi_1\land\psi_2$}{
     \Return $(\witness(\varphi)=\witness(\psi_1)\cup\witness(\psi_2))\land$\fun{$\lts,\witness, \psi_{1},\pi$} $\land$ \fun{$\lts,\witness, \psi_{2},\pi$}\;
   }
   \ElseIf{$\varphi=\fpinf\psi$}{
    \Return $(\exists 0\le i \le |\pi|, \forall i\le j \le i+\n, j\in\witness(\psi))\land$ \fun{$\lts,\witness,\psi,\pi$}\;
   }
 }
 \myFun{\funpath{$\lts, \pi$}}{
   \Return $\pi[0]$ is reachable from $\sinit$ and $\pi[|\pi|-1]$\;
   
 }
\end{algorithm}

\subsection{Fair model checking}

\thmfairMC*

We will need the following notions:

\begin{definition}[Strict predecessor]
  Given an LTS $\lts$ and a subset $U\subseteq\states$ of states,
  the strict predecessor set of $U$ is the set
  $\pre(U)=\{s\in\states\mid\successor(s)\subseteq U\}$.
\end{definition}

In other words, $\pre(U)$ describes the set of states $s$ from where
it is unavoidable to hit $U$ in one step.

\begin{definition}[Sure attractor]
  Given an LTS $\lts$ and a subset $U\subseteq\states$ of states,
  consider the operator $\mu \colon U\mapsto U\cup\pre(U)$.
  We call the smallest fixed point of $\mu$ denoted by $\pre^*(U)$ the set of
  \emph{sure attractor}.
\end{definition}

In other words, $\pre^*(U)$ describes the set of states $s$ from where it is unavoidable to hit $U$
at some point in the future. Note that one can state that it is the set of states from where
it is unavoidable to hit $U$ in $|S|$ or less steps, as otherwise, a loop that avoids $U$ can be found.

\begin{figure}
  \centering
      \begin{tikzpicture}[>=latex', join=bevel, thick, initial text =]
    \node[] (1) at (00bp, 0bp) [draw, circle] {$a$};
    \node[] (2) at (50bp, -30bp) [draw, circle] {$c$};
    \node[] (3) at (50bp, 30bp) [draw, circle] {$b$};
    \node[] (4) at (100bp, -30bp) [draw, circle] {$e$};
    \node[] (5) at (100bp, 30bp) [draw, circle] {$d$};
    \draw[->,] (1) [] to node [below] {} (2);
    \draw[->,] (1) [] to node [below] {} (3);
    \draw[->,] (2) [] to node [below] {} (3);
    \draw[->,] (2) [] to node [below] {} (4);
    \draw[->,] (3) [] to node [below] {} (5);
    \draw[->,] (5) [bend left] to node [below] {} (4);
    \draw[->,] (4) [bend left] to node [below] {} (5);
    \draw[->,] (2) [loop below] to node [below] {} (2);  
  \end{tikzpicture}
  \caption{\label{fig:sureattractor}System of Example~\ref{ex:sureattractor}}
\end{figure}

\begin{example}\label{ex:sureattractor}
 Consider the system depicted in \figurename~\ref{fig:sureattractor}, and let $U=\set{e}$. Then,
 $\pre(U)=\set{d}$, and $\pre^*(U)=\set{e,b,d}$. Note that $c\not\in\pre^*(U)$ because there is
 a self loop that ensures that from $c$, there will always be the possibility to stay in $c$ and
 therefore never reach $e$. Then, as there is the possibility from $a$ to reach $c$, we have that
$a\not\in\pre^*(U)$
\end{example}

Note that computing $\pre(U)$ can be done in polynomial time. Then, computing $\pre^*(U)$ consists
in at most $|\states|$ iterations of $\mu$, that is at most $|\states|$ computations of a strict predecessor,
and therefore can also be computed in polynomial time.

Since formulas in $\f\varphi\in\f(\fpLTLplus)$ define \emph{prefix
  independent} sets, we can ignore any finite prefix of a run, and
only check that $\varphi$ is almost surely satisfied in each bottom SCC (BSCC).  The
decomposition into BSCCs can be done using standard graph
techniques. The crux resides in checking whether a BSCC $B$ is
``good'', \textit{i.e.}~whether the formula is almost surely satisfied in $B$.
This is done by a \emph{tableaux like} approach. We recursively compute
the set of states from where \emph{sub-formulas} are satisfied. The key
observation lies in the fact thet $B$ is ``good'' iff the
formula can be satisfied from anywhere in $B$.  This intuition is
described in Algorithm~\ref{alg:fairMCprexind}.

\begin{algorithm}[H]
  \caption{Polynomial time algorithm for the fair model checking of
    $\TfpLTL$\label{alg:fairMCprexind}}
  \KwData{An LTS $\ltsFull$ and a formula $\f\varphi\in\TfpLTL$}
  \KwResult{Whether $\lts\asMod\f\varphi$}
  \SetKwFunction{algo}{CheckSystem}
  \SetKwFunction{fun}{$\induction$}
  \SetKwData{tmp}{tmp}
  \SetKwData{c}{C}
  \medskip
  \SetKwProg{myalg}{}{}{}
  \myalg{\algo{$\lts, \f\varphi$}}{
    \SetKwData{flag}{flag}
    \SetKwData{comp}{Components}
    $\comp\leftarrow$ The set of all the BSCCs of $\lts$\;
    \ForAll{$\c\in \comp$}{
      \If{$\fun{$\c,\varphi$} \neq\c$}{
        \Return {\sf no}\;
      }
    }
    \Return {\sf yes}\;
  }
\medskip
\SetKwProg{myFun}{}{}{}
\myFun{\fun{$\c, \varphi$}}{
  \If{$\varphi \in \ap$} {
    \Return the set $\{s\in \c \mid \varphi \in \lbl(s) \}$\;
  }
  \ElseIf{$\varphi=\psi\land\theta$}{
    \Return \fun{$\c,\psi$} $\cap$ \fun{$\c,\theta$}\;
  }
  \ElseIf{$\varphi=\psi\lor\theta$}{
    \Return \fun{$\c,\psi$} $\cup$ \fun{$\c,\theta$}\;
  }
  \ElseIf{$\varphi=\fpinf\psi$}{
    \If{$\pre^*(\fun{$\c,\psi$})=\c$}{
      \Return $\c$\;
    }
    \Else{
      \Return $\emptyset$\;
      }
    }
  }
\end{algorithm}

Thanks to Prop.~\ref{prop:fairMC} and Lem.~\ref{lem:polyBSCC},
Thm.~\ref{thm:fairMC} follows by
noticing that the set of all the BSCCs can be computed in polynomial time.

\begin{restatable}[]{proposition}{propfairMC}\label{prop:fairMC}
  Given an LTS $\lts$ and a formula $\f\varphi\in\f(\fpLTLplus)$, then
  $\lts \asMod \f\varphi$ iff ${\sf \algo}$ returns {\sf yes}.
\end{restatable}

In order to state the proof of Prop.~\ref{prop:fairMC}, we first need to
introduce some technical properties of Algorithm~\ref{alg:fairMCprexind}.


\begin{lemma}\label{lem:01BSCC}
  Given an LTS $\lts$, a BSCC $B$ of $\lts$ and a formula
  $\varphi\in\fpLTLplus$, then exactly one of the followings holds:
  \begin{itemize}
  \item ${\induction(B,\varphi)} = B$.
  \item ${\induction(B,\varphi)} = \emptyset$.
  \end{itemize}
\end{lemma}

\begin{proof}
  Let $B$ be a BSCC of $\lts$, and $\varphi\in\fpLTLplus$.  Let us prove
  the result by structural induction over $\varphi$. By definition of
  $\fpLTLplus$, three cases can occur.
  \begin{itemize}
  \item $\varphi=\fpinf\psi$, $\psi\in\fpLTL$: Then,
    ${\induction(B,\varphi)}$ returns either $B$ or $\emptyset$,
    and the result holds.
  \item $\varphi=\psi\lor\psi'$, $\psi,\theta\in\fpLTLplus$: By induction
    hypothesis, ${\induction(B,\psi)}$ is either equal to $B$
    or to $\emptyset$, and the same holds for
    ${\induction(B,\psi')}$.  Given that
    ${\induction(B,\varphi)}={\sf
      \induction(B,\psi)}\cup{\induction(B,\psi')}$,
    ${\induction(B,\varphi)}$ is either empty (if both
    ${\induction(B,\psi)}$ and
    ${\induction(B,\psi')}$ are empty) or equal to $B$
    otherwise.
   \item $\varphi=\psi\land\psi'$, $\psi,\theta\in\fpLTLplus$: By
    induction hypothesis, ${\induction(B,\psi)}$ is either
    equal to $B$ or to $\emptyset$, and the same holds for
    ${\induction(B,\psi')}$.  Given that
    ${\induction(B,\varphi)}={\sf
      \induction(B,\psi)}\cap{\induction(B,\psi')}$,
    ${\induction(B,\varphi)}$ is either equal to $B$ (if both
    ${\induction(B,\psi)}$ and
    ${\induction(B,\psi')}$ are equal to $B$) or empty
    otherwise.
  \end{itemize}
\end{proof}

\begin{lemma}\label{lem:mergingstartingcounterexample}
  Let $\psi\in\fpLTL$ be a formula, $N=|\states|+1$, $B$ be a BSCC and $\pi\in B^{N+1}$ a finite path such that for all
  $0\leq i\leq N$, there exists a run $\rho_i\in \pi[i]B^\omega$ such
    that $(\rho_i,N)\not\models\psi$.
    Then, there exists a run $\rho\in B^\omega$ such that $\rho\sub{0}{(N+1)}=\pi$ and for all $0\leq i\leq N$,
    $(\rho\sub{i}{},N)\not\models\psi$.
\end{lemma}

This lemma is a technical lemma that states that multiple counter examples can be merged into one.
It is used in the proof of the next lemma.

\begin{proof}
  Assume that $\pi\in B^{N+1}$ is a finite path such that for all
  $0\leq i\leq N$, there exists a run $\rho_i\in \pi[i]B^\omega$ such
    that $(\rho_i,N)\not\models\psi$.
    Let us construct the counter example $\rho$.
    Consider $\rho_i\in \pi[i]B^\omega$ such that $(\rho_i,N)\not\models\psi$.
    By Coro.~\ref{cor:ufeq}, there exists a finite prefix $w_i$ of $\rho_i$ such that any run $\rho'_i\in\cyl(w_i)$ satisfies
    $(\rho'_i,N)\not\models\psi$.
    Then, as $B$ is strongly connected, consider a finite path $x_0$ such that $(\pi[N],x_0[0])\in\tr$
    and $(x_0[|x_0|-1],w_0[0])\in\tr$,
    and for each $1\leq i\leq N$, consider a finite path $x_i$ such that $(w_{i-1}[|w_{i-1}|-1],x_i[0])\in\tr$
    and $(x_i[|x_i|-1],w_i[0])\in\tr$.
    Then, $\pi'=\pi x_0 w_0 x_1 w_1 \cdots x_N w_N$ is a finite path over $B$.
    Let $\rho$ be any run in $\cyl(\pi')$.
    Then, by construction, $\rho\sub{0}{(N+1)}=\pi$.
    Let us prove that for each $0\leq i\leq N$, $(\rho\sub{i}{},N)\not\models\psi$.
    Consider any $0\leq i\leq N$.
    By construction, there exists a $j\geq i$ such that $\rho\sub{i}{}=\rho\sub{i}{j} w_i\rho'$.
    As $w_i\rho'\in\cyl(w_i)$, $(w_i\rho')\not\models\psi$.
    As $\rho[i]=w_i[0]$, by Lem.~\ref{lem:pumpingloopstart}, $(\rho\sub{i}{},N)\not\models\psi$.
\end{proof}

\begin{lemma}\label{lem:correctioninduction}
  Given an LTS $\lts$, a BSCC $B$ of $\lts$, a formula
  $\varphi\in\fpLTL$, $b$ a state in $B$, and $N=|\states|+1$, then the
  following assertions are equivalent:
  \begin{itemize}
  \item $b \in {\induction(B,\varphi)} $.
  \item Any run $\rho$ in $bB^\omega$ satisfies
    $(\rho,N)\models\varphi$.
  \end{itemize}
\end{lemma}

\begin{proof}
  Let us show the equivalence by structural induction over $\varphi$.
  \begin{itemize}
  \item $\varphi=\theta\in\ap$: Then,
    $b \in {\induction(B,\theta)}$ iff
    $\theta\in\lbl(b)$.  As $\theta\in\lbl(b)$ iff any
    run $\rho$ in $bB^\omega$ satisfies
    $(\rho,N)\models\varphi$, the equivalence holds.
  \item $\varphi=\psi\land\psi'$: First, assume that
    $b \in {\induction(B,\varphi)} $.  By construction, this
    implies that $b \in {\induction(B,\psi)} $ and
    $b \in {\induction(B,\psi')} $.  By induction hypothesis,
    any run $\rho$ in $bB^\omega$ satisfies
    $(\rho,N)\models\psi$ and $(\rho,N)\models\psi'$, and therefore
    $\rho,N\models\varphi$.

    Now, assume that any run $\rho$ in $bB^\omega$ satisfies
    $(\rho,N)\models\varphi$.  By definition, this means that
    $(\rho,N)\models\psi$ and $(\rho,N)\models\psi'$.  By induction
    hypothesis, this implies that
    $b \in {\induction(B,\psi)} $ and
    $b \in {\induction(B,\psi')} $.  By construction, this
    implies that $b \in {\induction(B,\varphi)} $.

  \item $\varphi=\psi\lor\psi'$: First, assume that
    $b \in {\induction(B,\varphi)} $.  By construction, this
    implies that $b \in {\induction(B,\psi)} $ or
    $b \in {\induction(B,\psi')} $.  W.l.o.g., assume $b \in {\induction(B,\psi)} $.  By
    induction hypothesis, any run $\rho$ in $bB^\omega$
    satisfies $(\rho,N)\models\psi$ and therefore
    $(\rho,N)\models\varphi$.

    Now, assume that any run $\rho$ in $bB^\omega$ satisfies
    $(\rho,N)\models\varphi$.  By definition, this means that
    $(\rho,N)\models\psi$ or $(\rho,N)\models\psi'$.  W.l.o.g., assume that $(\rho,N)\models\psi$.  By induction
    hypothesis, this implies that
    $b \in {\induction(B,\psi)} $.  By construction, this
    implies that $b \in {\induction(B,\varphi)} $.

  \item $\varphi=\fpinf\psi$: First, note that by Lem.~\ref{lem:01BSCC},
  either ${\induction(B,\varphi)}=B$ or ${\induction(B,\varphi)}=\emptyset$.
  Let us handle both cases.
  
    First case, ${\induction(B,\varphi)}=B$.
    Then, one has to show that for any run $\rho\in B^\omega$, $(\rho,N)\models\fpinf\psi$.
    Let us introduce a new label $a^\psi$.
    Let us label the states in ${\induction(B,\psi)}$ with
    $a^\psi$, and only those states.
    By induction hypothesis over $\psi$, for any $b\in B$, $b\in{\induction(B,\varphi)}$ iff
    for all $\rho\in bB^\omega$, $(\rho,N)\models\psi$.
    By definition of $a^\psi$, the previous equivalence can be stated as follows: for any $b\in B$,
    $a^\psi\in\lbl(b)$ iff for all $\rho\in bB^\omega$, $(\rho,N)\models\psi$.
    Therefore, if for all $\rho\in B^\omega$, we have $(\rho,N)\models\fpinf a^\psi$, then
    for all $\rho\in B^\omega$, we have $(\rho,N)\models\fpinf\psi$.
    As we assumed that ${\induction(B,\varphi)}=B$, we have $\pre^*({\induction(B,\psi)} )=B$.
    Therefore, from any state $b\in B$, any finite path of lenght $N$ reaches a state of ${\induction(B,\psi)}$,
    that is a state with label $a^\psi$.
    Therefore, for all $\rho\in B^\omega$, we have $(\rho,N)\models\fpinf a^\psi$, and thus,
    for all $\rho\in B^\omega$, we have $(\rho,N)\models\varphi$.

    Second case, assume that ${\induction(B,\varphi)}=\emptyset$.
    Therefore, we must show that for all $b\in B$, there exists a run $\rho\in bB^\omega$ such that $(\rho,N)\not\models\varphi$.
    As $B$ is strongly connected, it is enough to show that there exists a run $\rho\in B^\omega$ such that
    $(\rho,N)\not\models\fpinf\psi$, as from any $b\neq\rho[0]$, we have a finite path $\pi$ from $b$ to $\rho[0]$ and
    $(\pi\rho,N)\not\models\fpinf\psi$.
    This is equivalent to showing that there exists a run $\rho\in B^\omega$ such that for all $0\leq j\leq N$,
    $(\rho\sub{j}{},N)\not\models\psi$.
    In the same way as in the first case, let us label the states in ${\induction(B,\psi)}$ with
    $A^\psi$, and only those states.
    By induction hypothesis, for each $b\in B$, $a^\psi\in\lbl(b)$ iff for every run $\rho\in bB^\omega$,
    $(\rho,N)\models\psi$.
    As ${\induction(B,\varphi)}=\emptyset$, we have that $\pre^*({\induction(B,\psi)} )\neq B$.
    Let $b_0\in B\setminus \pre^*({\induction(B,\psi)} )$.
    By definition of $a^\psi$, $b_0\not\in\pre^*(\set{b\in B\mid a^\psi\in\lbl(b)} )$.
    Therefore, there exists a finite path $\pi$ such that $\pi[0]=b_0$, $|\pi|=N+1$ and for all $0\leq i\leq N$,
    $a^\psi\not\in\lbl(\pi[i])$.
    Using the induction hypothesis, this implies that for each $0\leq i\leq N$, there exists a run $\rho_i\in \pi[i]B^\omega$ such
    that $(\rho_i,N)\not\models\psi$.
    By Lem.~\ref{lem:mergingstartingcounterexample}, there exists a run $\rho\in B^\omega$ such that $\rho\sub{0}{(N+1)}=\pi$
    and for all $0\leq i\leq N$, $(\rho\sub{i}{},N)\not\models\psi$.
    Therefore, $(\rho,N)\not\models\fpinf\psi$.\qedhere
  \end{itemize}
\end{proof}

\noindent We can now prove Prop.~\ref{prop:fairMC}.

\begin{proof}[Proof of Prop.~\ref{prop:fairMC}]
  First, assume that $\lts$
  fairly satisfies the formula, that is, there exists a $k$ such that
  $\prob_{\lts}(\set{\rho \in \runsinit \mid (\rho,k) \models \f\varphi}) = 1$.
  Consider a BSCC $B$ of $\lts$.
  Then, since $B$ is reachable,
  $\prob_{\lts}(\set{\rho \in \runsinit \mid \text{$\rho$ ends in $B$}}) \geq 0$.
  This implies that
  $\prob_{B}(\set{\rho \in B^\omega \mid (\rho,k) \models \f\varphi}) = 1$,
  and since $B$ is a BSCC,
  $\prob_{B}(\set{\rho \in B^\omega \mid (\rho,k) \models \varphi}) = 1$.
  By Coro.~\ref{cor:ufeq}, this implies that
  for all $\rho\in B^\omega$, $(\rho,k)\models\varphi$.
  Since this holds for any $\rho$ starting in any
  state of $B$, by lemma \ref{lem:correctioninduction}, this implies that
  ${\induction(B,\varphi)}=B$.
  Since this holds for any BSCC $B$ of $\lts$, by definition of
  ${\sf \algo}$, it returns {\sf yes}.

  Now, assume that $\lts$ does not
  fairly satisfy the formula, that is, for all $k$,
  $\prob_{\lts}(\set{\rho \in \runsinit \mid (\rho,k) \models \f\varphi}) < 1$.
  Let us fix a $k$.
  Note that \[\prob_{\lts}(\set{\rho \in \runsinit \mid \text{$\rho$ does not end in some BSCC $B$}}) =0\,.\]
  Therefore, there exists a BSCC $B_k$ of $\lts$ such that
  $\prob_{B_k}(\set{\rho \in B_k^\omega \mid (\rho,k) \models \f\varphi}) < 1$,
  thus there exists at least one $\rho_k\in B_k^\omega$ such that
  $(\rho_k,k)\not\models \f\varphi$, and therefore, in particular,
  $(\rho_k,k)\not\models \varphi$.
  By lemma \ref{lem:correctioninduction}, this means that
  $\rho[0] \not\in {\induction(B,\varphi)} $ and therefore, by lemma
  \ref{lem:01BSCC}, since $\varphi\in\fpLTLplus$,
  ${\induction(B,\varphi)}=\emptyset$.
  By definition of ${\sf \algo}$, it returns {\sf no}.
\end{proof}

Finally, we derive a polynomial time complexity bound for Algorithm~\ref{alg:fairMCprexind}.

\begin{restatable}[]{lemma}{polyBSCC}
  \label{lem:polyBSCC}
  Given an LTS $\lts$, a BSCC B of $\lts$ and a formula
  $\f\varphi\in\f(\fpLTLplus)$, the procedure
  ${\induction(B,\varphi)}$ runs  in time $O(|\varphi||B|^2)$.
\end{restatable}

\begin{proof}
  We show by structural induction over $\varphi$ that
  ${\induction(B,\varphi)}$ can be computed in time $|\varphi||B|^2$.
  \begin{itemize}
  \item $\varphi\in\ap$ is a state formula. Then, ${\induction(B,\varphi)}$
    is the set of states $s$ such that
    $\varphi\in\lbl(s)$.
    This can be computed in time $|B|\leq|B|^2$.
  \item $\varphi=\psi_1\land\psi_2$. Then,
    ${\induction(B,\varphi)}={\induction(B,\psi_1)}\cap{\induction(B,\psi_2)}$.
    By induction hypothesis, ${\induction(B,\psi_1)}$ can be computed
    in time $|\psi_1||B|^2$, and ${\induction(B,\psi_2)}$ can be
    computed in time $|\psi_2||B|^2$.
    Since ${\induction(B,\psi_1)}\subseteq B$ and
    ${\induction(B,\psi_2)}\subseteq B$, the intersection can be
    computed in time $|B|$.
    Therefore, ${\induction(B,\varphi)}$ can be computed in time
    $|\psi_1||B|^2+|\psi_2||B|^2+|B|\leq (|\psi_1|+|\psi_2|+1)|B|^2\leq |\varphi||B|^2$.
  \item $\varphi=\psi_1\lor\psi_2$. Then,
    ${\induction(B,\varphi)}={\induction(B,\psi_1)}\cup{\induction(B,\psi_2)}$.
    By induction hypothesis, ${\induction(B,\psi_1)}$ can be computed
    in time $|\psi_1||B|^2$, and ${\induction(B,\psi_2)}$ can be
    computed in time $|\psi_2||B|^2$.
    Since ${\induction(B,\psi_1)}\subseteq B$ and
    ${\induction(B,\psi_2)}\subseteq B$, the union can be
    computed in time $|B|$.
    Therefore, ${\induction(B,\varphi)}$ can be computed in time
    $|\psi_1||B|^2+|\psi_2||B|^2+|B|\leq (|\psi_1|+|\psi_2|+1)|B|^2\leq |\varphi||B|^2$.
  \item $\varphi=\fpinf\psi$. Then, to compute ${\induction(B,\varphi)}$,
    first, ${\induction(B,\psi)}$ is computed.
    By induction hypothesis, this is done in time $|\psi||B|^2$.
    Then, $\pre^*$ is computed. This is done in time $|B|^2$.
    Therefore, ${\induction(B,\varphi)}$ is computed in time
    $(|\psi|+1)|B|^2=|\varphi||B|^2$.\qedhere
  \end{itemize}
\end{proof}

\end{document}